\documentclass[12pt]{article}
\usepackage[letterpaper,margin=1in]{geometry}
\usepackage{amsmath,amssymb,amsthm,mathtools}
\usepackage{txfonts}
\usepackage{bm}
\usepackage{booktabs,array,tabularx,multirow}
\usepackage{enumitem}
\usepackage{graphicx}
\usepackage{url}
\usepackage[round,authoryear]{natbib}
\setcitestyle{aysep={}}
\usepackage[hidelinks,hyperfootnotes=false]{hyperref}
\usepackage{setspace}
\usepackage{float}
\makeatletter
\renewcommand{\maketitle}{%
  \begin{center}\normalfont\normalsize\bfseries\@title\par
  \normalfont\normalsize\@author\par\end{center}\@thanks%
  \vspace{0.5\baselineskip}
}
\makeatother
\renewenvironment{abstract}{%
  \begin{center}\normalfont\bfseries Abstract\end{center}%
  \normalsize\noindent
}{\par}

\hypersetup{
  pdftitle={Identifying and Measuring Treatment-Effect Heterogeneity across Sites through Within-Site Randomization},
  pdfauthor={Takahiro Hoshino},
  pdfsubject={site effect heterogeneity, explained and residual variation, hybrid randomization, inference, experimental design},
  pdfkeywords={cluster randomized trials, multisite experiments, covariate adjustment, measurement error, prediction risk},
  hypertexnames=false
}

\setlist[itemize]{label=\ensuremath{\bullet},leftmargin=2.0em,itemsep=0.15em,topsep=0.25em}
\setlist[enumerate]{leftmargin=2.25em,itemsep=0.15em,topsep=0.25em}
\newcolumntype{Y}{>{\raggedright\arraybackslash}X}
\newcolumntype{R}[1]{>{\raggedright\arraybackslash}p{#1}}

\newtheorem{theorem}{Theorem}
\newtheorem{proposition}{Proposition}
\newtheorem{corollary}{Corollary}

\theoremstyle{definition}
\newtheorem{assumption}{Assumption}

\theoremstyle{remark}

\newcommand{\E}{\mathbb{E}}

\newcommand{\Var}{\operatorname{Var}}
\newcommand{\Cov}{\operatorname{Cov}}

\newcommand{\tr}{\operatorname{tr}}

\newcommand{\diag}{\operatorname{diag}}

\newcommand{\one}{\mathbf{1}}

\newcommand{\trans}{\mathsf{T}}

\title{Identifying and Measuring Treatment-Effect Heterogeneity across Sites\\ through Within-Site Randomization}
\author{Takahiro Hoshino\thanks{\protect\normalsize\protect\setstretch{1}Corresponding author: 2-15-45 Mita, Minato-ku, Tokyo 108-8345, Japan Email: \href{mailto:bayesian@keio.jp}{bayesian@keio.jp}.}\\Keio University and RIKEN, Tokyo, Japan}
\date{}
\begin{document}
\maketitle
\begin{abstract}
Parallel cluster-randomized trials can identify mean effects by site characteristics and risk differences between fixed predictors, while leaving residual heterogeneity and absolute prediction risk (mean squared error of site-effect predictions) unidentified. We acquire missing same-site cross-arm information by randomizing participants within selected sites. Under randomized site selection with compatible participant reference distributions and intervention versions, we characterize necessary and sufficient observation patterns for quadratic targets in a model with unrestricted joint distributions of site response means. From finite independent participant samples, we construct unbiased moments and exact covariances. Under regularity conditions, we prove consistency as informative paired site counts grow, allowing fixed within-site sample sizes. For a fixed predictor, we optimize paired-site selection and unequal cell counts using corrected risk variances, pilot error, and costs. In an International Stroke Trial information-deletion experiment, pilot allocation reduces the conditional variance of reconstructing the complete-information corrected risk estimate by 34.3\% relative to uniform allocation with equal expected cost. In 40 Honduras villages, correction reduces estimated residual variation from 0.04081 to 0.00282 in squared log-consumption units for the same fixed explanatory basis; the explained share rises despite a smaller explained component, illustrating how participant sampling noise can inflate estimated residual variation.
\end{abstract}
\section{Introduction}
\label{sec:intro}
In parallel cluster-randomized trials, researchers can compare mean squared errors of fixed site-effect predictors and learn how average effects vary with site characteristics, yet remain unable to measure absolute prediction risk or variation among similar sites. This residual variation is the mean squared error of the population conditional-mean predictor. We determine which observations identify these targets, how to measure them with finite participant samples, and how to acquire observations under a fixed budget.

The missing information is the product of the same site's treatment and control response means. A parallel cluster-randomized trial observes each site under only one arm. Without restrictions on the joint distribution of site response means, site populations can therefore have identical arm-specific response distributions and mean effects but different residual variation, even with arbitrarily many participants and single-arm sites. Randomizing independent participants to different arms within selected sites recovers the missing products. We use known positive selection probabilities to extend these measurements to the target site population. We require compatible participant reference distributions and intervention versions across experimental patterns.

To acquire cross-arm information within a study's budget, we add within-site comparisons to a prespecified site network for an individually delivered intervention with unchanged content across assignment patterns. We account for additional staff and participant measurements through the incremental cost $c_n(w)=c_0(w)+c_1(w)n$, combining paired-site and participant costs. We include full pairing as a candidate when feasible and use positive pairing probabilities across target contexts when resources support only partial pairing. We choose sites and participant counts while holding the treatment contrast and target population fixed.

\citet{DingFellerMiratrix2019} decompose individual treatment-effect variation and bound its residual component. \citet{Hines2025} use identified prediction-risk differences while an absolute-risk term remains unidentified. \citet{Lu2023} standardize participant covariates for site comparisons. \citet{LiLuedtke2023} and \citet{LiGilbertDuanLuedtke2025} develop efficient estimation by fusing aligned and weakly aligned sources. We determine which same-site arm combinations the experiment must create and connect their observation requirements to corrected measurement and allocation for a specified quadratic target.

We establish three connected results. First, Theorem~\ref{thm:identification} characterizes necessary and sufficient observation patterns for quadratic targets in a model with unrestricted joint distributions and identified projected laws, including multiple arms and subgroups. For coordinate observation patterns, the condition specifies which arm-by-subgroup means must be observed together within a site. Second, we construct unbiased moments from finite independent within-site samples (Proposition~\ref{prop:pair}). We also obtain exact covariances for unequal arm-by-reference cells under finite second response moments (Supplement Proposition~S1). Under Theorem~\ref{thm:rates}'s regularity conditions, we can consistently estimate heterogeneity as informative paired sites accumulate even with fixed within-site sample sizes. The convergence rate for estimating explained variation depends on total site counts, whereas the rate for residual variation depends on paired site counts. A Gaussian submodel gives an information lower bound showing why more single-arm observations cannot replace the required paired information.

Third, we derive variance-minimizing paired-site probabilities and participant counts for a specified corrected risk estimator (Theorems~\ref{thm:allocation} and \ref{thm:jointdesign}), including unequal cell-count menus in the supplement. We base cost-adjusted allocation on measurement noise, between-site variation in squared prediction errors, pilot estimation error, and costs \citep{Neyman1934}. We use an independent pilot to fix the predictor and estimate these inputs before evaluation. Sites with the greatest predicted benefits need not be those most informative about residual error.

We evaluate corrected measurement and allocation in simulations and two applications. Simulations compare estimation and allocation against known population targets (Section~\ref{sec:simulation}). In the International Stroke Trial (IST), we evaluate placement of paired information within an observed hospital frame; in Honduras, we separate fixed-basis village-effect variation from finite-household noise.

In IST, we use 75 pilot hospitals to fix the predictor and allocation before evaluation in 151 hospitals with a common age reference (Section~\ref{sec:ist_application}). Pilot allocation favors small and middle-size hospitals, where corrected risk features have greater estimated variance. Across 1,000 information deletions, pilot allocation reduces conditional reconstruction variance by 34.26\% and root mean squared reconstruction error by 17.75\% relative to uniform allocation. Both designs reconstruct the complete-information corrected risk estimate in that fixed frame, with expected 45.3 paired evaluation hospitals, the same pilot cost, and fixed participant arrays. The budget covers expected incremental pairing costs; the gains quantify the precision of reconstructing the corrected risk estimate for the unchanged predictor.

In Honduras, we study the within-village treatment contrast in the same 40 intervention villages using a fixed basis of recorded population and distance (Section~\ref{sec:hondurasapplication}). Under independent household response sampling and nonselective outcome availability within village and arm, correction reduces the residual estimate from 0.04081 to 0.00282 in squared log-consumption units. Although the explained component also decreases, the point estimate of the explained share rises from 14.0\% to 62.1\% without changing the basis. We estimate finite-village projections using each village's participant distribution and report component intervals to quantify their substantial uncertainty. We use existing paired data in Honduras to evaluate corrected measurement and the IST hospital frame to evaluate where to acquire paired measurements more precisely.

\section{Framework and Identification}
\subsection{Site-level treatment effects}
\label{sec:setup}
Sites $j=1,\ldots,J$ are sampled from a target site population $Q$. Within site $j$, $Y_{ij}(a)$ is the potential outcome under individual treatment $a\in\{0,1\}$ and $\bm{X}_{ij}$ is a baseline covariate vector. We describe heterogeneity using the recorded site characteristic $W_j$, measured before assignment of the randomization pattern. Unless otherwise stated, expectations below are over $Q$, so each site receives equal weight. Known transport weights can instead define a different target population of sites.

For a specified participant reference distribution $P^*$, we define
\begin{equation}
 \theta_{aj}^*=\int m_{aj}(\bm{x})\,dP^*(\bm{x}),\qquad
 m_{aj}(\bm{x})=\E\{Y_{ij}(a)\mid \bm{X}_{ij}=\bm{x},j\},\qquad
 \tau_j=\theta_{1j}^*-\theta_{0j}^*.
 \label{eq:thetaag}
\end{equation}
Using a common reference, we standardize recorded participant characteristics while retaining institutional context. We can instead specify the same site-specific reference $\Pi_j$ in both arms before assignment of the randomization pattern, with $\theta_{aj}^{\Pi}=\int m_{aj}\,d\Pi_j$. When $\Pi_j=P_j$, the site's participant population law, the treatment effect refers to that site's own participant distribution. We call this a native-composition effect; its heterogeneity includes composition and context. The same moment-recovery argument applies under the conditional feature, independence, and pattern-invariance conditions (Supplement Proposition~S7). In Honduras, we study native-composition effects using recorded village attributes. In IST, we use a common age reference and reverse the sign of \eqref{eq:thetaag} so benefit is risk reduction.

We use treatment randomization to identify the causal contrast in \eqref{eq:thetaag}. Conditioning on $W$, we compare these treatment effects across existing settings under the specified participant reference. We use the recorded distribution of $W$ in the target site population to summarize this variation.

Write $m(W)=\E(\tau\mid W)$, $v(W)=\Var(\tau\mid W)$, and
\begin{equation}
 \mu=\E\tau,\qquad B=\Var\{m(W)\},\qquad
 R=\E\{v(W)\},\qquad V=\Var(\tau)=B+R.
 \label{eq:decomp}
\end{equation}
We use $B$ to measure variation in the conditional mean treatment effect explained by the recorded characteristics, and $R$ to measure variation among sites with the same characteristics. For a specified predictor $f$, we define its absolute prediction risk as
\begin{equation}
 \mathcal R(f)=\E\{(\tau-f(W))^2\}
 =R+\E\{(m(W)-f(W))^2\}.
 \label{eq:risk}
\end{equation}
The residual component is therefore the prediction risk of the population conditional-mean predictor. For a fitted predictor, we also account for the error in estimating that conditional mean. We report the proportion $B/V$ as a secondary summary because it becomes unstable when $V$ is close to zero.

With a limited number of sites, we use a fixed basis $\bm{b}(W)\in\mathbb R^p$ of interpretable functions to define finite-dimensional projection targets. The basis includes an intercept and satisfies $\bm{G}=\E\{\bm{b}\bm{b}^\trans\}\succ0$. Let
\begin{equation}
 \bm{t}=\E(\bm{b}\tau),\quad \bm{\beta}=\bm{G}^{-1}\bm{t},\quad f_b=\bm{b}^\trans\bm{\beta},
 \quad B_b=\bm{t}^\trans \bm{G}^{-1}\bm{t}-\mu^2,\quad
 R_b=\E(\tau^2)-\bm{t}^\trans \bm{G}^{-1}\bm{t}.
 \label{eq:projectiontargets}
\end{equation}
These quantities measure explained variation and residual risk for the best predictor in the chosen span. They satisfy $V=B_b+R_b$ and
\begin{equation}
 B-B_b=R_b-R=\E\{(m-f_b)^2\}\ge0.
 \label{eq:approxgap}
\end{equation}
We allow the true conditional mean to be unrestricted and define its projection onto the chosen span. For a finite partition $W\in\{1,\ldots,H\}$ and its indicator basis, $f_b(h)=m_h$ and $R_b=\sum_hq_hv_h$, where $q_h=\Pr(W=h)$. We use this partition formulation for the IST regional decomposition. For Honduras, we report the finite-site projection for a fixed village-attribute basis; for the IST design evaluation, we target the absolute risk of a fixed pilot predictor.

\subsection{Experimental design}
Let $C_j=1$ denote a site assigned within-site randomization, with known probability $\rho_J(W_j)$. Both arm means can then be estimated in that site. A site with $C_j=0$ is assigned wholly to arm $A_j$, with $\Pr(A_j=1\mid W_j,C_j=0)=\pi(W_j)$. Let $I_{aj}$ indicate that arm $a$ is observed; its marginal probability is
\begin{equation}
 p_{aJ}(w)=\rho_J(w)+\{1-\rho_J(w)\}\pi_a(w),\qquad
 \pi_1=\pi,\quad\pi_0=1-\pi.
 \label{eq:armprob}
\end{equation}
We call the sites with $C=1$ paired sites. Within-site randomization must be operationally feasible, and its intervention version must agree with that in whole-site assignment for the primary hybrid estimand. For example, this requirement is plausible for a patient-level treatment whose content is unchanged by the randomization level. A design for an intervention with substantial within-site spillovers needs a specified spillover or version model.

\begin{assumption}[Randomization, support, and site features]
\label{ass:main}
Sites are independent. Conditional on $W$, the experimental pattern is randomized independently of the site response population. Treatment is randomized with known probabilities within paired sites; consistency and absence of interference between sites hold. The reference $P^*$ is supported in every included site. Available first-moment features $S_a$ and second-moment features $U_{ab}$ satisfy
\[
 \E(S_a\mid\theta_0^*,\theta_1^*,W,I_a=1)=\theta_a^*,\qquad
 \E(U_{ab}\mid\theta_0^*,\theta_1^*,W,\text{pair }ab\text{ observed})
       =\theta_a^*\theta_b^*.
\]
The same target moments apply to each pattern in which a feature is used. Approximate features are permitted under the aggregate error conditions in the supplement.
\end{assumption}

\subsection{Identification}
\label{sec:identification}
First moments identify mean treatment effects conditional on site characteristics and their explained variation. To identify residual variation, we also need the cross-arm second moment from the same site. We characterize the observation patterns required for these targets in the following theorem.

\begin{theorem}[Information required for heterogeneity]
\label{thm:identification}
Under Assumption~\ref{ass:main} and finite second moments, suppose both arms have positive observation probability conditional on $W$.
\begin{enumerate}[label=(\roman*)]
\item A pure cluster trial identifies $m(W)$, $B$, and $(\bm{\beta},B_b)$. It also identifies $\mathcal R(f)-\mathcal R(g)$ for fixed square-integrable $f,g$.
\item Write $m_a(W)=\E(\theta_a^*\mid W)$, $v_a(W)=\Var(\theta_a^*\mid W)$, and $c(W)=\Cov(\theta_0^*,\theta_1^*\mid W)$. Then
\[
 R=\E\{v_0(W)+v_1(W)-2c(W)\}.
\]
Arm-specific second-moment features identify $v_0,v_1$. The additional cross-arm moment identifies $R$, $R_b$, and absolute prediction risk. These quantities generally remain unidentified without paired information in the unrestricted model.
\item If $\rho_J(W)>0$ almost surely and paired sites supply $U_{01}$, the missing second moment is identified by inverse-probability weighting. Finite within-site samples suffice when the displayed feature equations hold.
\item More generally, fix $W=w$ and let pattern $d$ reveal $\bm{M}_d\bm{\Theta}^*$, where $\bm{\Theta}^*\in\mathbb R^K$ collects arm-by-subgroup means. If each projected law is identified, then, over the unrestricted finite-second-moment model, for symmetric $\bm{A}$, the functional $\E(\bm{\Theta}^{*\trans}\bm{A}\bm{\Theta}^*\mid w)$ is identified if and only if
\[
 \bm{A}\in\operatorname{span}\{\bm{M}_d^\trans \bm{B}_d\bm{M}_d:\bm{B}_d=\bm{B}_d^\trans,
                         \ \Pr(d\mid w)>0\}.
\]
For coordinate selectors on an interior rectangular support, this is equivalent to observing every nonzero coordinate pair of $\bm{A}$ together in some available pattern, under the bounded-density model stated in the supplement.
\end{enumerate}
\end{theorem}

We express residual variation and recover the missing cross-arm moment as follows:
\begin{align}
 R&=\E(\theta_1^{*2}+\theta_0^{*2}-2\theta_0^*\theta_1^*)-\E\{m(W)^2\},
 \label{eq:missingmoment}\\
 \E\{CU_{01}/\rho_J(W)\mid W\}
 &=\E(\theta_0^*\theta_1^*\mid W).
 \label{eq:pairingrecovery}
\end{align}
We set the contribution in \eqref{eq:pairingrecovery} to zero when $C=0$ and estimate this moment by averaging observed cross-arm products over paired sites with inverse-probability weights. To see why the product matters, take constant $W$ and $U\sim\mathrm{Bernoulli}(1/2)$. The populations $(\theta_0^*,\theta_1^*)=(U,U)$ and $(U,1-U)$ have identical arm-specific laws and zero mean effect, but residual variances zero and one. They can produce the same pure-cluster records even with arbitrarily many participants per site. We distinguish these populations using products from different participants assigned different arms in the same site. We thereby recover the cross-arm moment of site response means; the product of one individual's two potential outcomes is a distinct estimand.

Risk differences depend only on first moments because
\begin{equation}
 \mathcal R(f)-\mathcal R(g)
 =\E\{f^2-g^2-2m(f-g)\}.
 \label{eq:riskdifference}
\end{equation}
As discussed by \citet{Hines2025}, the cancellation in \eqref{eq:riskdifference} permits comparison of two fixed predictors using only first moments. To estimate the absolute risk, we acquire the cross-arm moment through pairing. Table~\ref{tab:requiredpairs} summarizes the observations needed for several arm contrasts.

\begin{table}[htbp]
\centering\small
\caption{Information required beyond arm means and marginal second moments. Pairs are observed within the same site; all required patterns have positive probability.}
\label{tab:requiredpairs}
\begin{tabularx}{\textwidth}{Y R{0.30\textwidth}}
\toprule
Target & Additional cross-arm pairs\\
\midrule
$B$, $B_b$, or $\mathcal R(f)-\mathcal R(g)$ & None\\
$R$, $R_b$, or $\mathcal R(f)$ for $\tau=\theta_1^*-\theta_0^*$ & $01$\\
$\Var(\theta_1^*-\theta_0^*)$ & $01$\\
$\Var(\theta_2^*-\theta_0^*)$ & $02$\\
$\Cov(\theta_1^*-\theta_0^*,\theta_2^*-\theta_0^*)$ & $01$, $02$, and $12$\\
\bottomrule
\end{tabularx}
\end{table}

For the last row, the cross-moment is $\E(\theta_1^*\theta_2^*-\theta_0^*\theta_1^*-\theta_0^*\theta_2^*+\theta_0^{*2})$. When we observe $01$ and $02$ at different sites, we still require same-site observations of $12$ in the unrestricted model. In part (iv) and Supplement Section~S5, we derive the necessary and sufficient pattern condition, including subgroups. We first require positive observation probabilities for the necessary patterns, then choose an allocation to determine estimation precision. Our targets are moments of the site response means, which we recover from finite within-site samples under the feature conditions.

We prove Theorem~\ref{thm:identification} in Supplement Section~S1. For the general pattern condition, we recover the observable quadratic forms and construct site populations with identical projected laws but different target values whenever a required pattern is absent.

\subsection{Residual heterogeneity when pairing is feasible only in some contexts}
\label{sec:limitedpairing}
Some institutions can implement both arms, whereas others can implement only whole-site assignment. Let $\mathcal E$ be the set of contexts in which pairing has positive probability. The target site population remains $Q$, including its complement $\mathcal E^c$. Single-arm observations can identify the explained component throughout that population, whereas residual variation in $\mathcal E^c$ is partially identified.

For the sharp bounds, we consider an oracle experiment that identifies the conditional marginal distributions $F_a(\cdot\mid w)$ of the standardized site means. We assume this information is available either from observed site means or from an asymptotic experiment with sufficiently accurate site means. With finite participant samples, the two identified moments support the outer bounds described below; sharp quantile bounds additionally use the conditional marginal laws. Let $q_a(u,w)$ be their generalized quantile functions and assume finite second moments. Define
\begin{align}
 K_+(w)&=\int_0^1q_0(u,w)q_1(u,w)\,du,\nonumber\\
 K_-(w)&=\int_0^1q_0(u,w)q_1(1-u,w)\,du,\nonumber\\
 r_-(w)&=\E(\theta_0^{*2}+\theta_1^{*2}\mid w)-2K_+(w)-m(w)^2,\nonumber\\
 r_+(w)&=\E(\theta_0^{*2}+\theta_1^{*2}\mid w)-2K_-(w)-m(w)^2.
 \label{eq:quantilevariance}
\end{align}
These expressions bound the conditional variance of the site-effect difference by rearranging the marginal quantiles. The rearrangement argument is classical; see \citet{AronowGreenLee2014} for its use in randomization-variance bounds. We use it to bound the residual heterogeneity outside the pairing support.

\begin{theorem}[Identification with incomplete pairing support]
\label{thm:incomplete}
Suppose the conditional marginal laws above are identified throughout $Q$, the conditional joint law is identified on $\mathcal E$, and the model permits all measurable conditional couplings of the marginals on $\mathcal E^c$. Then the sharp identified set for $R$ is
\begin{equation}
 \left[\E\{\one_{\mathcal E}v(W)+\one_{\mathcal E^c}r_-(W)\},\
       \E\{\one_{\mathcal E}v(W)+\one_{\mathcal E^c}r_+(W)\}\right].
 \label{eq:partialR}
\end{equation}
The sets for $V$ and $\mathcal R(f)$ are obtained by adding, respectively, the identified quantities $B$ and $\|m-f\|_2^2$. The interval width is
\begin{equation}
 \mathcal W(\mathcal E)=2\E\{\one_{\mathcal E^c}[K_+(W)-K_-(W)]\}.
 \label{eq:identwidth}
\end{equation}
It vanishes exactly when $K_+=K_-$ almost surely on the unpaired contexts.
\end{theorem}

We construct the endpoints using common and oppositely ordered conditional uniform variables. A mixture of the two couplings attains every intermediate scalar value while preserving all observed marginal laws. Using positive pairing probabilities on $\mathcal E$, we identify its contribution even when those probabilities differ greatly. We improve estimation precision by increasing those probabilities and change the identified set itself by expanding $\mathcal E$. With finite lower-level samples, we recover the moment targets in Theorem~\ref{thm:identification}; the sharp bounds in \eqref{eq:quantilevariance} additionally require the marginal quantiles. Where only conditional variances $v_0,v_1$ are available, $r_-$ and $r_+$ can be replaced by $(\sqrt{v_1}-\sqrt{v_0})^2$ and $(\sqrt{v_1}+\sqrt{v_0})^2$ to obtain valid outer bounds. Additional support and marginal information can tighten these outer bounds.

We use \eqref{eq:identwidth} to distinguish the choice of pairing contexts from the allocation of paired sites within them. If permitting pairing in a context has a fixed institutional cost, its reduction in identification width is $2[K_+-K_-]$. For finitely many context groups, we can choose the groups through a cost-constrained subset optimization with those weighted gains. Within the selected groups, we allocate paired sites using the precision criterion in Section~\ref{sec:design}. A pairing probability close to zero can identify the population target while leaving little information for estimation in a finite experiment.

\paragraph{Changes in intervention delivery.}
Whole-site and within-site assignment can change intervention delivery and thereby change the site effect being estimated. In Supplement Theorem~S2, we derive a joint sensitivity region for the mean and heterogeneity components, together with a bound for fixed-predictor risk, under a root-mean-square bound on delivery differences. For the primary hybrid estimand, we assume compatible intervention versions.

\section{Estimation and Inference}
\label{sec:inference}
\subsection{Correcting within-site sampling variation}
The expected squared site-effect estimate includes both the effect's second moment and the estimate's sampling variance. We remove the sampling-variance contribution using products from distinct independent evaluation units and use the corrected moment features to estimate the heterogeneity components.

\begin{proposition}[Products from independent evaluation units]
\label{prop:pair}
Conditional on the site and an independent training sample, let $\bm{S}_{jb}\in\mathbb R^K$, $b=1,\ldots,N_j$, be independent evaluation-unit scores with common mean $\bm{\Theta}_j^*$. For $N_j\ge2$,
\begin{equation}
 U_{k\ell,j}=\frac{1}{\binom{N_j}{2}}\sum_{b<b'}
 \frac{S_{kjb}S_{\ell jb'}+S_{\ell jb}S_{kjb'}}2
 \label{eq:pair}
\end{equation}
has conditional mean $\theta_{kj}^*\theta_{\ell j}^*$, including when $k=\ell$. The evaluation units can be independent subclusters; their distributions need not be identical if their score means agree.
\end{proposition}

For independent observations with sample mean $\bar Y_a$, we use the same-arm feature $\{(\sum_iY_{ia})^2-\sum_iY_{ia}^2\}/\{n_a(n_a-1)\}$. For independent arm samples, we use $\bar Y_0\bar Y_1$ as the cross-arm feature. Combining these products gives the corrected feature for the squared site effect,
\begin{equation}
 Q_j=U_{11,j}+U_{00,j}-2U_{01,j},\qquad
 \widetilde\tau_j=S_{1j}-S_{0j},\qquad
 \E(Q_j\mid\bm{\Theta}_j^*,W_j,C_j=1)=\tau_j^2.
 \label{eq:Q}
\end{equation}
For a predictor $f$ fixed by independent training, we construct the corrected risk feature from the same products:
\begin{equation}
 L_f=Q-2f(W)\widetilde\tau+f(W)^2,
 \qquad \E(L_f\mid j,C=1)=\{\tau_j-f(W_j)\}^2.
 \label{eq:pairedrisk}
\end{equation}
Its population mean is $\mathcal R(f)$; $f=m$ gives $R$, and $f=f_b$ gives $R_b$. We retain signed realizations when averaging this feature. Supplement Proposition~S1 gives its exact variance for unequal arm-by-reference cells. With at least four observations per cell, Proposition~S2 constructs unbiased covariance estimators using distinct observations. Section~\ref{sec:jointdesign} uses measurement variance to choose paired sites and participant counts.

For fitted scores, we train each product's score functions on observations outside its two evaluation units. To assess product bias, we account for both conditional mean errors and covariances between scores. Shared fitted reference distributions or regressions can create covariance even between scores based on different arms. In Supplement Section~S2, we derive the bias identity and state the conditions for transferring the corrected features to the site-level analysis. For IST, we construct the features entirely from cell means and distinct-patient products with fixed age weights.

\subsection{Augmented moment signals}
We estimate $g_a(W)$ to predict $\theta_a^*$ and $g_2(W)$ to predict $\tau^2$, using separate sites or site-level cross-fitting. With the known design probabilities, we define
\begin{align}
 H_{aj}&=g_a(W_j)+\frac{I_{aj}}{p_{aJ}(W_j)}\{S_{aj}-g_a(W_j)\},\nonumber\\
 Z_{1j}&=H_{1j}-H_{0j},\qquad
 Z_{2j}=g_2(W_j)+\frac{C_j}{\rho_J(W_j)}\{Q_j-g_2(W_j)\}.
 \label{eq:signals}
\end{align}
We set the contribution of an unobserved residual to zero and train the predictors using observed outcomes from the training sites. Conditional on independent training, these signals satisfy $\E(Z_1\mid W)=m(W)$ and $\E(Z_2\mid W)=\E(\tau^2\mid W)$ regardless of the accuracy of $g_a,g_2$. We retain the lower-level sampling variation in $S_a,Q$ when calculating the site-level covariance. Thus $Z_2-2f(W)Z_1+f(W)^2$ has mean $\mathcal R(f)$ for a fixed independently trained $f$. Additional marginal second moments can serve as control variates; any design using them must optimize the covariance of that full target score.

For a cross-fitted mean predictor $\widehat f$, we define
\begin{align}
 \widehat\mu&=\mathbb P_J Z_1,\nonumber\\
 \widehat B&=\mathbb P_J(2\widehat fZ_1-\widehat f^2)-\widehat\mu^2,\qquad
 \widehat R=\mathbb P_J(Z_2-2\widehat fZ_1+\widehat f^2),\nonumber\\
 \widehat V&=\mathbb P_JZ_2-\widehat\mu^2=\widehat B+\widehat R.
 \label{eq:scoreest}
\end{align}
Here $\mathbb P_J$ denotes the average over evaluation sites. For a fitted $f$ independent of an evaluation site,
\begin{equation}
 \E(2fZ_1-f^2)-\mu^2=B-\|f-m\|_2^2,\quad
 \E(Z_2-2fZ_1+f^2)=R+\|f-m\|_2^2.
 \label{eq:learningerror}
\end{equation}
These expectation identities show that prediction error reduces the explained component captured by the predictor relative to the population $B$ and increases its residual risk. We therefore interpret a small estimate in light of the predictor's accuracy. For fits in the fixed span of $\bm{b}$, the corresponding prediction error is $\|f-f_b\|_2^2$ and the targets are $B_b,R_b$.

\begin{theorem}[Joint inference at the two information rates]
\label{thm:rates}
Let $\rho_J(w)=r_J\bar\rho(w)$, with $0<c\le\bar\rho(w)\le C<\infty$ and $Jr_J\to\infty$. Consider independent oracle signal records satisfying the conditional equations above, $\widehat\mu-\mu=O_p(J^{-1/2})$, and define
\[
 \phi_B=2(m-\mu)(Z_1-m)+(m-\mu)^2-B,
 \qquad \phi_R=Z_2-2mZ_1+m^2-R.
\]
Suppose the triangular array $(\phi_B,\sqrt{r_J}\phi_R)$ satisfies a multivariate Lindeberg condition, its covariance converges to $\bm{\Sigma}\succ0$, and $\E\{(Z_1-m)^2\mid W\}$ is uniformly bounded. On a fixed number of site folds suppose
\[
 \|\widehat f-m\|_4=o_p(1),\qquad
 \|\widehat f-m\|_2^2=o_p(J^{-1/2}).
\]
Then
\begin{equation}
 \begin{pmatrix}\sqrt J(\widehat B-B)\\ \sqrt{Jr_J}(\widehat R-R)\end{pmatrix}
 =\frac1{\sqrt J}\sum_{j=1}^J
 \begin{pmatrix}\phi_{Bj}\\\sqrt{r_J}\phi_{Rj}\end{pmatrix}+o_p(1)
 \ \rightsquigarrow\ N(0,\bm{\Sigma}).
 \label{eq:rates}
\end{equation}
For $R$ alone the squared prediction-error requirement is $o_p\{(Jr_J)^{-1/2}\}$. For a fixed span, replace $m,B,R$ by $f_b,B_b,R_b$; sufficient additional conditions are root-$J$ consistency of each training-fold coefficient estimate, $\E\|\bm{b}\|^4<\infty$, and $\E\{\|\bm{b}\|^2(Z_1-f_b)^2\}<\infty$. Estimated signals give the same limit if their scaled aggregate remainders vanish. Consistent covariance estimation additionally requires vanishing empirical squared error of the estimated scaled influence vectors and a law of large numbers for their oracle second moments.
\end{theorem}

The theorem gives different information rates for the explained and residual components: $\sqrt J$ for the former and $\sqrt{Jr_J}$ for the latter. Positive pairing probability identifies the residual component, and a growing number of paired sites permits consistent estimation. Under the exact feature and moment conditions, these results allow finite lower-level counts. For the fixed-span targets, we allow a nonlinear true conditional mean and estimate its projection on the specified basis. In the reduced site-feature model, we derive a local minimax lower bound using the following Gaussian submodel with known marginal laws. The feature-error conditions specify when the site-level inference applies to individual-level hierarchical records.

\paragraph{Information bound for residual variation.}
\label{sec:lowerbound}
For constant $W$, consider an oracle Gaussian site vector with known marginal means, positive standard deviations $\sigma_0,\sigma_1$, and unknown correlation $\gamma\in(-1,1)$. Single-arm records carry no information about $\gamma$; each paired record has Fisher information $I_1(\gamma)=(1+\gamma^2)/(1-\gamma^2)^2$. The local asymptotic minimax risk for $R=V=\sigma_0^2+\sigma_1^2-2\gamma\sigma_0\sigma_1$, on the $Jr_J$ scale, is at least $4\sigma_0^2\sigma_1^2/I_1(\gamma_0)$. Thus the rate in Theorem~\ref{thm:rates} cannot be improved uniformly over a model containing this submodel. If $Jr_J$ remains bounded, uniform consistency on a fixed nondegenerate interior neighborhood is impossible. Supplement Section~S3 gives the likelihood expansion and two-point argument.

\subsection{Finite strata and boundary inference}
For finitely many strata, we estimate the components directly using unbiased products across independent sites. Conditional on stratum counts $J_h\ge2$, we fix the target masses at $q_h=J_h/J$, write $\bar Z_{1h},\bar Z_{2h}$ for the stratum means, and let $s_{1h}^2$ be the sample variance of $Z_1$ in stratum $h$. We define
\begin{equation}
 M_{2h}=\bar Z_{1h}^2-s_{1h}^2/J_h,\quad
 \widehat\mu=\sum_hq_h\bar Z_{1h},\quad
 M_{2\mu}=\widehat\mu^2-\sum_hq_h^2s_{1h}^2/J_h.
 \label{eq:finitemeans}
\end{equation}

\begin{proposition}[Exact stratum decomposition]
\label{prop:finite}
Conditional on independent pilot fits and the stratum counts, suppose site signal records are independent and identically distributed within each stratum. Then
\begin{equation}
 \widehat B_b=\sum_hq_hM_{2h}-M_{2\mu},\quad
 \widehat R_b=\sum_hq_h(\bar Z_{2h}-M_{2h}),\quad
 \widehat V=\sum_hq_h\bar Z_{2h}-M_{2\mu}
 \label{eq:finiteest}
\end{equation}
are unbiased for the corresponding fixed-mixture targets and satisfy $\widehat B_b+\widehat R_b=\widehat V$ exactly.
\end{proposition}

These estimators are exactly unbiased conditional on the pilot under independent sampling of sites. For predictors fitted on overlapping folds, we use the asymptotic result. Conditioning on the observed stratum counts fixes the mixture weights and permits independent sampling of sites within strata; inference for a random population composition also includes variation in estimated $q_h$.

We retain signed estimates in tables and impose nonnegative parameter constraints through joint confidence regions. We first construct a region for the first and second moments $(m_h,q_{2h})$, where $q_{2h}=\E(\tau^2\mid h)$, intersect it with $q_{2h}\ge m_h^2$, and project it onto the heterogeneity components. We obtain coverage for a zero explained component by projecting the moment region, including cases where the normal limit for $\widehat B_b$ is degenerate at zero. For the binary IST outcomes, simultaneous Bonferroni intervals for the $2H$ primitive means produce a conservative joint rectangle, intersected with $-1\le m_h\le1$ and $m_h^2\le q_{2h}\le1$. In the supplement, we give explicit bounds for $B_b,R_b,V$ and the group variances. Their asymptotic coverage follows from the simultaneous intervals for the moments. For unbounded continuous outcomes, we retain $q_{2h}\ge m_h^2$ and use a confidence region with the outcome's appropriate support; the constraint $q_{2h}\le1$ applies to the binary-outcome analysis. If the intersection is empty, we report the full feasible range and record that event. We assess precision through the confidence region and interpret a signed estimate below zero as part of this uncertainty assessment rather than as evidence of identical site effects.

\subsection{Inference near zero explained variation}
\label{sec:weakexplanation}
At a constant conditional mean, the first derivative of explained variation vanishes. We then obtain a centered quadratic limit for the finite-stratum estimator, reflecting the vanishing first-order derivative. In Supplement Theorem~S3, we derive the zero and local limits and the mean-contrast test. We project a mean-contrast confidence ellipsoid onto explained variation to obtain intervals covering zero and positive values. For the explained fraction, we project a joint moment region after specifying a positive minimum total variation; permitting arbitrarily small total variation can yield an uninformative interval for this ratio. We report signed component estimates and use the asymptotic coverage and testing results. Calibration when strata contain few sites requires a separate finite-sample analysis.

When we estimate a common participant reference, its shared estimation error contributes to the covariance. In Supplement Section~S11.1, we derive this contribution and the conditions for combining it with the site influence expansion. We fix the age mixture for IST and use the specified native participant law for Honduras.

\section{Optimal Allocation}
\label{sec:design}
\subsection{Paired-site selection}
We fix the site population, participant reference, and predictor $f$ before comparing allocations. For the feature in \eqref{eq:pairedrisk}, we define $a_f(W)=\E(L_f\mid W)$ and $\nu_f^2(W)=\Var(L_f\mid W)$. This variance includes sampling variation in arm means and their products from finite within-site samples. We minimize the variance of the estimator of $\mathcal R(f)$ across allocations with the same predictor and risk target. We construct the direct augmented estimator by averaging
\begin{equation}
 a_f(W)+\frac{C}{\rho(W)}\{L_f-a_f(W)\}.
 \label{eq:riskAIPW}
\end{equation}
We can use the same construction for a finite vector of group risks, with design criteria based on the estimator covariance. The optimization over observation patterns in the supplement also accommodates marginal control variates and joint targets.

\begin{theorem}[Conditional allocation for absolute prediction risk]
\label{thm:allocation}
Condition on independent training and pilot data, fix $f$, and suppose $L_f$ has finite second moment. Let $C\mid(W,L_f)$ be Bernoulli with known probability $\rho(W)>0$. The oracle-augmented record in \eqref{eq:riskAIPW} has mean $\mathcal R(f)$ and variance
\begin{equation}
 \Var\{a_f(W)\}+\E\{\nu_f^2(W)/\rho(W)\}.
 \label{eq:riskvariance}
\end{equation}
Let $c(W)>0$ be the additional cost of a paired site. Over $\ell(W)\le\rho(W)\le u(W)$, with $0<\ell\le u\le1$, and $\E\{c(W)\rho(W)\}\le b_0$, an interior feasible, binding-budget optimum is
\begin{equation}
 \rho^*(W)=\min\left[u(W),\max\left\{\ell(W),
       \sqrt{\frac{\nu_f^2(W)}{\lambda c(W)}}\right\}\right],
 \label{eq:allocation}
\end{equation}
where $\lambda>0$ enforces the budget. At a zero conditional variance the lower bound is optimal. If the pilot supplies a fixed predictor $\widetilde a_f$, replace $a_f$ by $\widetilde a_f$ in the record \eqref{eq:riskAIPW}. Its mean remains $\mathcal R(f)$ and its variance, conditional on the pilot, is
\begin{equation}
 \Var(a_f)+\E(\nu_f^2/\rho)
   +\E\bigl[\{\rho^{-1}-1\}(\widetilde a_f-a_f)^2\bigr].
 \label{eq:pilotpenalty}
\end{equation}
Thus the variance-minimizing allocation for that fixed predictor replaces $\nu_f^2$ in \eqref{eq:allocation} by $\nu_f^2+(\widetilde a_f-a_f)^2$.
\end{theorem}

Under \eqref{eq:allocation}, we prioritize sites according to the conditional variance of the corrected risk feature relative to cost. This variance can be large because of residual effect variation or small within-site samples. Consequently, sites with the largest predicted treatment benefits need not be the most informative sites for estimating residual variation. For multiple risk targets with positive-semidefinite importance matrix $\bm{H}$, replace $\nu_f^2$ by the conditional trace $\tr\{\bm{H}\Var(\bm{L}\mid W)\}$. A regional target requires positive pairing probability in every region whose residual variation is to be estimated.

For the corrected residual-risk feature, variance minimization gives the cost-adjusted square-root allocation principle \citep{Neyman1934}. We determine the allocation inputs and eligibility restrictions from the hybrid experiment and the risk estimand. We set the lower probability bound to support inference as well as allocation: positive probabilities identify the target, consistency requires $Jq_h\rho_h\to\infty$ for each target stratum, and the resulting covariance determines finite-sample precision. With $J_h$ eligible sites in stratum $h$, the expected paired count is $J_h\rho_h$. To implement fixed counts, we sample sites at random within strata. Conditional finite-frame variances then have finite-population corrections, whereas inference for a population of sites includes sampling variation of the frame itself. We give both formulas in the supplement.

To calculate the allocation, we use second moments of the full corrected risk feature, including variation in its latent quadratic terms. A cross-arm covariance bound constrains only part of this variation. We can estimate conditional variances from an independent paired pilot and fix both the allocation and the risk predictor before observing the evaluation sites. Alternatively, known outcome and weight bounds provide conservative upper bounds, or we can optimize the worst-case variance over a pilot confidence set. In the numerical evaluation, we use estimated pilot inputs to assess the implementable rule.

For a fixed predictor, we quantify the variance contribution of error in the augmentation function exactly through \eqref{eq:pilotpenalty}. If $f$ is fitted to estimate the best projection, its squared prediction error must also meet the rate in Theorem~\ref{thm:rates} for inference about $R_b$. An independent evaluation can instead estimate $\mathcal R(\widehat f_{\rm pilot})$ conditional on the pilot without that convergence requirement. The IST design evaluation uses this absolute prediction risk as its target.

To implement the design, we first specify the site population, reference, recorded characteristics, and residual or absolute-risk target. We then identify the required arm pairs, the feasible randomization patterns, and the independent units from which to construct products. Using an independent paired pilot, we fix $f$, the augmentation functions, and the allocation inputs, then choose participant counts and pairing probabilities under cost and support constraints. We estimate the target using the corrected features and recorded probabilities and construct a confidence set to quantify precision. We next derive the role of participant counts in this allocation.

\subsection{Joint site and participant allocation}
\label{sec:jointdesign}
We choose paired sites and participant counts using the variance of the corrected risk feature. We decompose this variance into within-site sampling variation, which decreases with participant counts, and between-site variation in squared prediction errors.

Condition on independent training and pilot data, and fix $f$. Conditional on a site, let $T_{j1},\ldots,T_{jn}$ be independent and identically distributed effect scores with mean $\tau_j$ and variance $s_j^2$. We choose $n$ using $W$ and the pilot before observing evaluation outcomes. For example, under Bernoulli individual randomization with probability $e_j$, the score $AY/e_j-(1-A)Y/(1-e_j)$ has this mean for a native-composition effect. Under the feature conditions, we can use known reference weights or independently trained augmentation to construct a standardized score. We may take an independently randomized subcluster as the lower-level unit. We define
\begin{equation}
 L_{f,n}=\frac1{n(n-1)}\sum_{i\ne k}(T_{ji}-f(W_j))(T_{jk}-f(W_j)).
 \label{eq:Ln}
\end{equation}
We obtain \eqref{eq:pairedrisk} with $\widetilde\tau_n=n^{-1}\sum_iT_{ji}$ and $Q_n=\{n(n-1)\}^{-1}\sum_{i\ne k}T_{ji}T_{jk}$: its conditional mean is $\{\tau_j-f(W_j)\}^2$ for every $n$. We use the stated i.i.d. condition to derive the participant-count formula for general score distributions, including non-Gaussian scores; the general feature equations require fewer conditions. We write
\begin{align}
 A_f(w)&=\Var\{(\tau-f(w))^2\mid W=w\},\nonumber\\
 D_f(w)&=4\E\{(\tau-f(w))^2s_j^2\mid W=w\},\qquad
 E_f(w)=2\E(s_j^4\mid W=w).
 \label{eq:ADE}
\end{align}
Here $s_j^4$ denotes the square of a conditional variance.

\begin{theorem}[Information per site and joint resource allocation]
\label{thm:jointdesign}
Assume the moments in \eqref{eq:ADE} and the displayed design criteria are finite, and $n\ge2$. Then
\begin{equation}
 \nu_{f,n}^2(w)=\Var(L_{f,n}\mid W=w)
 =A_f(w)+D_f(w)/n+E_f(w)/\{n(n-1)\}.
 \label{eq:exact_n_variance}
\end{equation}
Let $\widetilde a_f(w)$ be one fixed pilot augmentation common to all candidate counts, let $a_f(w)=\E\{(\tau-f(w))^2\mid W=w\}$, and put
\[
 \eta_{f,n}^2(w)=\nu_{f,n}^2(w)+\{a_f(w)-\widetilde a_f(w)\}^2.
\]
For the actual pilot-augmented record, \eqref{eq:pilotpenalty} makes the design-dependent variance
\[
 \E\{\eta_{f,n(W)}^2(W)/\rho(W)\}.
\]
Oracle augmentation is the case $\widetilde a_f=a_f$. Let $c_n(w)=c_0(w)+c_1(w)n$, with positive coefficients, and let $n$ range over a prespecified finite set $\mathcal N(w)$. Subject to $\E\{c_{n(W)}(W)\rho(W)\}\le b_0$ and the stated probability bounds, for any $\lambda>0$ minimize
\begin{equation}
 \eta_{f,n}^2(w)/r+\lambda c_n(w)r
 \label{eq:jointlagrange}
\end{equation}
over $n\in\mathcal N(w)$ and $\ell(w)\le r\le u(w)$ at each $w$. If these choices satisfy the budget with equality, they are globally optimal. For each candidate $n$, the minimizing $r$ is the clipped square-root rule. If the minimizing rules are interior for every candidate $n$, the optimal participant count minimizes
\begin{equation}
 n^*(w)\in\arg\min_{n\in\mathcal N(w)}
       c_n(w)\eta_{f,n}^2(w).
 \label{eq:bestn}
\end{equation}
\end{theorem}

We obtain the variance formula from the order-two product decomposition \citep{Hoeffding1948}. Its linear part is $2(\tau-f)n^{-1}\sum_i(T_i-\tau)$; its degenerate part averages $(T_i-\tau)(T_k-\tau)$ over distinct units. Using their zero covariance, we obtain \eqref{eq:exact_n_variance}. By increasing $n$, we reduce the two measurement terms while retaining $A_f$ and the fixed pilot augmentation error. If $A_f+(a_f-\widetilde a_f)^2>0$, the product in \eqref{eq:bestn} eventually increases with $n$ and has a finite integer minimizer when $c_1>0$. If this term is zero, the preferred count may reach the allowed maximum.

To compute the allocation, we enumerate feasible participant counts and solve for the multiplier. When integer choices switch, we can randomize between tied count--probability pairs to attain an intermediate expected budget and record the selected pair. A fixed realized budget requires the fixed-count implementation in the supplement. When probability bounds bind, we optimize \eqref{eq:jointlagrange} directly; the boundary example in Supplement Section~S12.5 shows why clipping can change the preferred count. Supplement Proposition~S6 extends the optimization to unequal arm-by-reference cell counts at a common target.

We use a common augmentation across participant counts because a count-specific function would make the subtracted squared error in \eqref{eq:pilotpenalty} depend on $n$ and change the joint objective. The optimum uses true inputs conditional on the pilot; for estimated pilot inputs, we use the pilot-error guarantees in Supplement Section~S6. We can estimate the between-site component and the two measurement components by subsampling participant counts within the paired pilot sites. We use the exact variance identity for expected-cost optimization. Finite-sample interval coverage and compliance with a fixed realized budget require separate analyses.

For several regional risks or outcomes, we optimize allocation using the covariance of the corrected features. In Supplement Theorem~S4, we characterize the allocation that minimizes the variance of the least precise standardized contrast and give a pilot-error bound for the vector of targets.

\section{Simulation Studies}
\label{sec:simulation}
We examine mean-effect modification, residual-variance modification, and prognostic prediction using three equally frequent site strata with arm sizes 80, 40, and 20. Three bounded-response mechanisms share mean site effect .04 and site-effect variance .00827; a fourth has constant effects. We estimate augmentation functions from an independent 60-site-per-stratum pilot. We evaluate $J=300$ or 900 and pairing probabilities .2 or .5 in 1,000 replications per setting, correcting within-site sampling variation by distinct-response products. Supplement Section~S12 gives generating equations, seeds, and comparisons.

The corrected estimators distinguish explained and residual components in mechanisms with identical total site-effect variance. In the residual-variance mechanism, $B=0$ while conditional residual variances are .0016, .0064, and .0168. Uncorrected squared site-effect estimates overestimate $R$ by about .014 at $J=300$, pairing probability .2, exceeding the true residual component. In the mean-modifier setting at this evaluation size, complete-pair estimation has lower residual RMSE than fusion. Single-arm information improves estimation of the explained component; its contribution to precision depends on the target.

\begin{figure}[htbp]
\centering
\includegraphics[width=.90\textwidth]{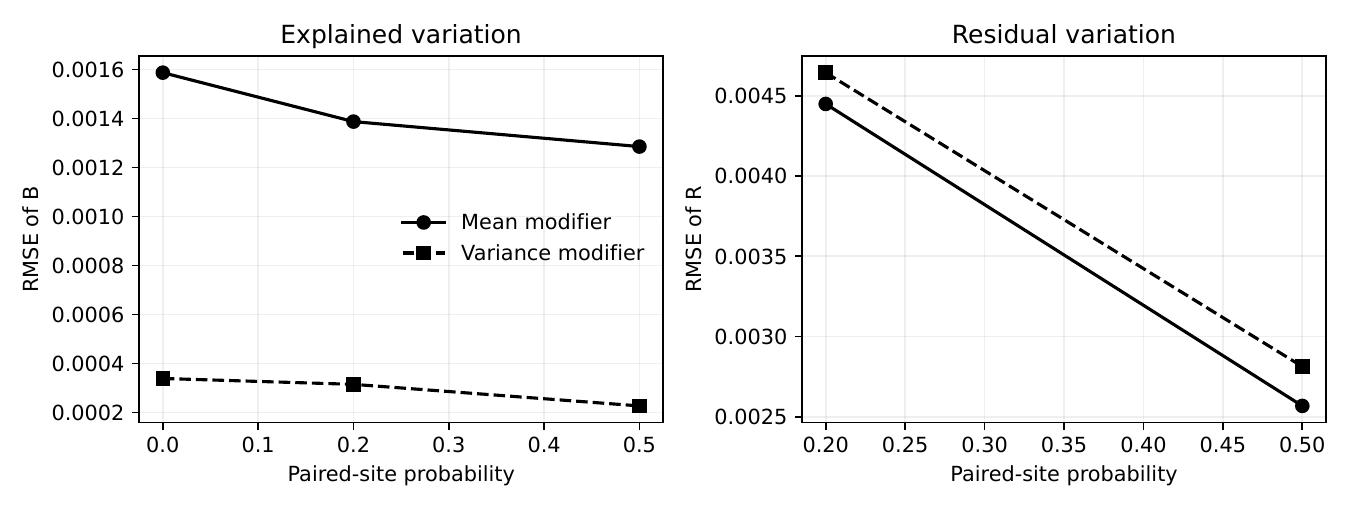}
\caption{Simulation RMSE at $J=300$. Whole-site data identify the explained component; residual variation requires paired information. Lines join only evaluated settings.}
\label{fig:newrates}
\end{figure}

Increasing the pairing probability or the number of sites reduces residual RMSE in these simulations. In the residual-variance mechanism, residual RMSE is .00464 at $(J,\rho)=(300,.2)$, .00281 at $(300,.5)$, and .00264 at $(900,.2)$. Coverage of residual Wald intervals is .918--.928 in the four $(300,.2)$ settings, and explained-variance Wald inference is nonregular at zero. Projecting the confidence region for the first and second moments gives joint coverage .991--1.000 in a separate grid, with conservative widths (Supplement Table~S1).

We conduct a separate allocation experiment with 600 evaluation sites, a 180-site pilot, expected pairing fraction .30, and bounds [.10,.90]. We evaluate the same fixed pilot predictor under both allocations and include the same pilot cost. Pilot allocation reduces the absolute-risk estimator’s RMSE from .00255 to .00222 at the common budget (Table~\ref{tab:designeval}). Supplement Section~S12.5 gives a numerical check of joint site--participant allocation under a common measurement budget. 

\section{Applications}
\label{sec:applications}
\subsection{International Stroke Trial}
\label{sec:ist_application}
The International Stroke Trial randomized 19,435 patients to aspirin or avoidance, crossed with heparin assignment \citep{IST1997,ISTdata,Sandercock2011}. We use the within-hospital trial records to delete one arm and assess where paired measurements improve absolute-risk estimation. We study aspirin benefit---risk reduction in 14-day mortality and six-month death/dependency---averaged over the trial's factorial heparin-assignment law. In the allocation experiment for absolute prediction risk of 14-day benefit, we fix the predictor, hospital frame, participant reference, and expected pairing budget.

\paragraph{Common reference and explanatory basis.}
We use the same age mixture for both arms and endpoints: 75\% age 80 or younger and 25\% older than 80, retaining other characteristics' within-age distributions and weighting eligible hospitals equally. We require at least two observed responses per arm-by-age cell for both endpoints, retaining 226 of 466 public hospital identifiers and 16,047 baseline patients. Eligibility uses availability counts without response values; complete-case means assume response-independent availability within hospital, arm, and age cell. We use United Kingdom, other Europe, and other-country indicators to describe how hospital effects vary across settings.

\paragraph{The missing information and its recovery.}
When we retain one aspirin arm per evaluation hospital, we preserve arm-specific information and remove the cross-arm product $U_{01,j}=S_{0j}S_{1j}$. Theorem~\ref{thm:identification} establishes absolute-risk nonidentification in the unrestricted site-population model. In the deletion experiment, we fix the observed frame and evaluate how precisely paired measurements reconstruct its complete-information corrected estimate. We fix prediction and design inputs from the pilot and estimate the evaluation frame's second moment from its own retained records. Paired hospitals provide the missing product and distinct-patient squares $U_{00,j},U_{11,j}$. With benefit score $\widetilde\tau_j=S_{0j}-S_{1j}$, the theory's corrected feature is
\[
 Q_j=U_{00,j}+U_{11,j}-2U_{01,j},\qquad
 L_{f,j}=Q_j-2f(W_j)\widetilde\tau_j+f(W_j)^2.
\]
We calculate the cell means and products directly from disjoint patients using fixed age weights. Under independent patient response sampling, we recover the stated target moments with finite cell counts, taking expectations over patient responses rather than treatment assignment alone in a fixed patient roster. Supplement Section~S13 gives the exact products, outcome coding, availability assumptions, and allocation protocol.

\paragraph{An independent pilot selects the paired hospitals.}
We select 75 paired pilot hospitals by an outcome-independent stratified random draw, leaving 151 for evaluation. We fix the regional benefit predictor from the pilot and use baseline enrollment counts 38 and 77 to define measurement-precision strata. We use region to explain hospital effects and enrollment strata to estimate the precision of corrected squared-error measurements.

We fix augmentation and allocation inputs from the pilot's $L_f$. With one incremental cost unit per paired hospital, the probabilities are .401, .403, and .100 in small, middle, and large hospitals (Supplement Table~S6), versus uniform .30. The pilot variance estimates for corrected risk features are smaller in larger hospitals; we retain their information through the .10 probability floor. We choose pairing probabilities to improve the precision of the absolute-risk estimator for the fixed benefit predictor.

In each of 1,000 deletions, we retain both arms at paired hospitals and one randomly selected aspirin arm elsewhere. Hidden outcomes are removed before feature computation and cannot reconstruct missing cross-arm products. Both allocations use the same 151 evaluation hospitals, fixed predictor and augmentation, corrected risk feature, and expected 45.3 paired evaluation hospitals, and include the same 75-hospital pilot cost. We measure the budget by expected incremental hospital cost, holding participant arrays fixed and allowing realized paired counts to vary; recruitment costs are excluded.

\begin{table}[htbp]
\centering\small
\caption{Acquiring information for the same fixed pilot predictor's absolute risk. The single-arm row denotes unavailable same-hospital cross-arm products; no RMSE is assigned. Nonidentification in the unrestricted site-population model is established in Theorem~\ref{thm:identification}. Numeric rows compare placement at the same expected paired count within each experiment (1,000 replications). Simulation errors use known risk conditional on the pilot; IST errors reconstruct the full-information corrected estimate in the fixed evaluation frame without pilot extrapolation. Bias and RMSE have squared risk-difference units. Mean pairs are realized averages; expected counts are 180 and 45.3, respectively.}
\label{tab:designeval}
\begin{tabular}{llrrr}
\toprule
Evaluation & Allocation & Bias & RMSE & Mean pairs \\
\midrule
IST single-arm & No paired evaluation & \multicolumn{2}{c}{Products unavailable} & 0 \\
\midrule
Simulation & Uniform & 0.00004 & 0.00255 & 179.7 \\
Simulation & Pilot allocation & 0.00003 & 0.00222 & 180.1 \\
IST deletion & Uniform & 0.00005 & 0.00419 & 45.0 \\
IST deletion & Pilot allocation & 0.00013 & 0.00345 & 45.2 \\
\bottomrule
\end{tabular}

\end{table}

We use the full-information corrected risk estimate, .00322, as the IST reference conditional on the fixed evaluation frame and pilot. With pilot allocation, deletion RMSE falls from .004192 to .003448 (17.75\%); exact conditional variance falls from $1.711\times10^{-5}$ to $1.125\times10^{-5}$ (34.26\%). Both deletion estimators are conditionally unbiased for this reference; the small empirical biases in Table~\ref{tab:designeval} are Monte Carlo averages. These deletion errors measure reconstruction uncertainty for the estimated reference conditional on the frame and pilot. By changing pairing probabilities, we change the uncertainty of estimating $\mathcal R(f)$ while holding $f$ and its true population risk fixed.

\paragraph{Estimating regional and residual hospital-effect variation.}
We treat the 226 hospitals as independent response-population draws within region, conditional on stratum counts 70, 129, and 27 and weights $J_h/226$. Applying Proposition~\ref{prop:finite}'s finite-stratum correction, we estimate mean benefits of $-.0008$ and $.0174$ (Table~\ref{tab:istdecomp}). For 14-day mortality, the regional component intervals have upper limits .00127 for $B_b$ and .00957 for $R_b$. We retain the negative signed regional estimate and six-month variance components and assess heterogeneity through intervals that admit both small and larger residual differences. For the best regional projection, we quantify uncertainty over response-population sampling of hospitals rather than exact treatment randomization for a fixed vector of 226 hospital effects; for deletion, we quantify reconstruction uncertainty for the fixed pilot predictor's risk estimate conditional on the realized 151-hospital frame.

\begin{table}[htbp]
\centering\small
\caption{Aspirin benefit across the same 226 hospitals. Components are signed noise-corrected estimates, in squared risk-difference units. Brackets are nominal simultaneous 95\% intervals obtained by projecting a primitive moment confidence region within each endpoint, with asymptotic coverage from the primitive approximation and no joint claim across endpoints.}
\label{tab:istdecomp}
\begin{tabular}{lrrrr}
\toprule
Outcome & Mean & $B_b$ & $R_b$ & $V$ \\
\midrule
14-day mortality & $-0.0008$ & $-0.00003$ & $0.00261$ & $0.00258$ \\
95\% interval & & $[0,0.00127]$ & $[0,0.00957]$ & $[0,0.00957]$ \\
6-month composite & $0.0174$ & $-0.00018$ & $-0.00117$ & $-0.00135$ \\
95\% interval & & $[0,0.00230]$ & $[0,0.00901]$ & $[0,0.00901]$ \\
\bottomrule
\end{tabular}

\end{table}

We report regional residual estimates and their upper limits in Supplement Table~S5. When we change the participant reference, we compare variation in different targets and interpret their variance difference as a target comparison rather than an additive composition contribution. We retain negative corrected estimates and assess the intervals using the asymptotic normal approximation for the primitive moments, with particular care for the small geographic stratum.

\begin{figure}[htbp]
\centering
\includegraphics[width=\textwidth]{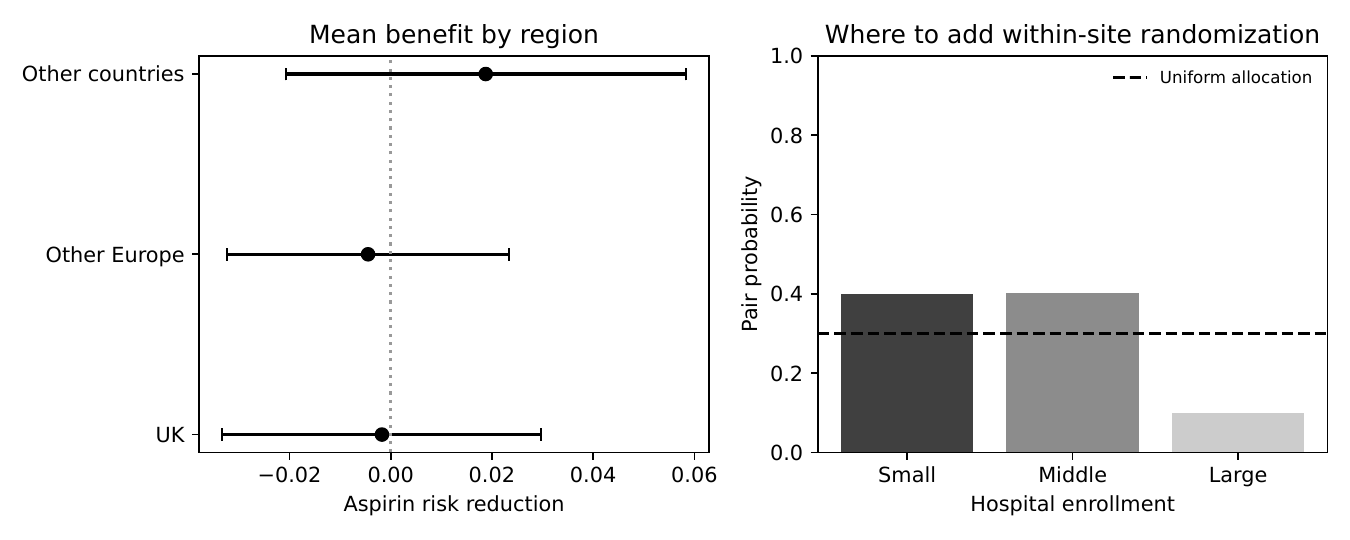}
\caption{The two roles of hospital information. Left: regional mean-benefit estimates and simultaneous intervals for the primitive moments of 14-day mortality in the 226-hospital frame. Right: enrollment-stratum pairing probabilities for measuring the fixed pilot predictor's absolute error in the 151 evaluation hospitals. The allocations change measurement precision, not the predictor or its true risk.}
\label{fig:istdesign}
\end{figure}

\subsection{Honduras Graduation Trial}
\label{sec:hondurasapplication}
We analyze the Honduras component of the Graduation study in \citet{Banerjee2015} to estimate explained and residual effect variation after removing finite-participant noise. Eligible households in 40 intervention villages were randomized to a package of productive assets, training, support, and coaching. We compare assigned and unassigned households within the intervention villages, retaining their trial environment and excluding the 40 pure-control villages from this contrast. The outcome is endline log monthly per-capita consumption. Of 1,609 households, 1,529 have an observed positive outcome: 767 assigned and 762 unassigned. All 40 intervention villages remain, with at least seven observations in each arm. We weight villages equally and use the same native participant-population reference in both arms, including participant composition and village context in effect variation.

\paragraph{Village-effect projections onto recorded attributes.}
We fix the explanatory basis---an intercept, log recorded village population, and the recorded binary distance indicator---and village weights across uncorrected and corrected measurements. We use recorded village population and distance as village descriptors; neither variable is constructed from household outcomes. Let $\bm{\tau}_F$ be the vector of the 40 native village effects, $\bm{P}$ the projection onto this three-dimensional basis, and $\bm{P}_0$ the intercept projection. The targets are
\[
 B_F=J^{-1}\bm{\tau}_F^\trans(\bm{P}-\bm{P}_0)\bm{\tau}_F,\qquad
 R_F=J^{-1}\bm{\tau}_F^\trans(\bm{I}-\bm{P})\bm{\tau}_F,\qquad V_F=B_F+R_F.
\]
We summarize the village-effect vector by its finite-frame projection onto the recorded basis. The targets allow departures from a linear effect model and differ from unrestricted conditional-mean explanation and random-effect variance. We fix the recorded descriptors retrospectively and interpret the coefficients as associations with existing village characteristics; their original measurement timing remains unverified (Section~S13.2).

\paragraph{Correcting finite-household sampling noise.}
For $Z_j=\bar Y_{1j}-\bar Y_{0j}$, the distinct-household products supply $Q_j=Z_j^2-\widehat v_j$, where $\widehat v_j=s_{1j}^2/n_{1j}+s_{0j}^2/n_{0j}$. Under independent household response sampling and nonselective outcome availability within village and arm, we use these native first- and second-moment features to average over household response sampling. We therefore interpret the correction within the response-population model, whose expectations differ from those obtained by fixing household potential outcomes and averaging only treatment randomization. Proposition~S8 corrects projection noise by subtracting $J^{-1}\tr\{(\bm{P}-\bm{P}_0)\bm{\widehat\Omega}\}$ from the explained quadratic and $J^{-1}\tr\{(\bm{I}-\bm{P})\bm{\widehat\Omega}\}$ from the residual quadratic, with $\bm{\widehat\Omega}=\operatorname{diag}(\widehat v_j)$. Conditional on this frame and basis, the corrected point estimators are unbiased for $B_F$ and $R_F$ without a Gaussian outcome assumption.

\begin{table}[htbp]\centering\small
\caption{The same 40-village frame and additive basis before and after finite-household correction. The uncorrected column projects observed village differences, which include sampling noise; the corrected column measures variation in the village-effect vector. Components have squared log-consumption units. Working bounds are approximate simultaneous 95\% containing bounds from a normal confidence region for the primitive moments under the response-population conditions. They are neither finite-sample guarantees nor confidence bounds for the ratio.}
\label{tab:hondurasmain}
\begin{tabular}{lrrr}
\toprule
Target & Uncorrected & Corrected & Working bounds \\
\midrule
$B_F$: specified explanation & 0.00664 & 0.00463 & [0.00001, 0.03318] \\
$R_F$: remaining variation & 0.04081 & 0.00282 & [0.00000, 0.05784] \\
$V_F=B_F+R_F$ & 0.04745 & 0.00745 & [0.00001, 0.05786] \\
$B_F/V_F$ (point ratio only) & 0.140 & 0.621 & --- \\
\bottomrule
\end{tabular}

\end{table}

Holding the village frame and basis fixed, we estimate residual variation of .04081 before correction and .00282 after correction (Table~\ref{tab:hondurasmain}). Under the stated response model, we attribute most apparent residual dispersion to finite-household sampling noise. The total estimate falls from .04745 to .00745 and the explained component from .00664 to .00463, while its point share rises from 14.0\% to 62.1\%.

We estimate the additive association with log population as $-.0707$ log-consumption units per unit log population, with working 95\% interval $[-.1497,.0084]$, and the distance-indicator association as $.1560$, with interval $[.0104,.3016]$. The positive distance coefficient indicates larger village effects for the recorded category coded one, conditional on log population. We interpret the 62.1\% point explanation ratio together with component bounds that include zero residual variation and allow substantial uncertainty. When we add the population-by-distance interaction, we obtain 49.5\%, with greater leverage and similarly weak precision. Supplement Section~S13.3 assesses the working intervals at these village sizes under normal and nonnormal responses. Section~S13.2 gives the protocol and sensitivity analyses. We interpret the working ranges under this fixed reanalysis specification, with earlier application and endline screening unadjusted.

In Honduras, we estimate explained and residual effect variation; in IST, we select hospitals to measure a fixed pilot predictor's absolute risk. In both analyses, paired products, equivalently sampling-variance corrections, recover the same-site cross-arm moments required for measurement and acquisition design.

\section{Discussion}
\label{sec:discussion}
We identified the same-site observations required for corrected measurement of residual heterogeneity and absolute prediction risk and derived allocation rules for acquiring them under a specified budget. Researchers can compare fixed predictors using risk differences even when their absolute risks remain unidentified. In the unrestricted whole-site model, increasing the number of single-arm sites or participants leaves the missing cross-arm information unresolved. Randomizing participants within selected sites supplies that information for corrected inference and resource allocation.

In IST, we improved reconstruction precision while fixing the predictor, age reference, hospital frame, and expected cost by allocating more paired measurements to small and middle-size hospitals with greater estimated variance of corrected risk features. In Honduras, we used both within-village arms to measure fixed-basis explained and residual variation for each village's participant population. Sampling correction reduced both components, but reduced residual variation more sharply, raising the explained share for the same basis. We report this share as a descriptive point estimate and summarize uncertainty with component intervals.

Under the feature conditions for the participant response-population model, we derived exact moment, covariance, and finite-projection corrections. We can consistently estimate population heterogeneity even when individual site-effect estimates remain noisy, provided informative paired site counts grow under the regularity conditions for consistency and asymptotic inference; within-site sample sizes need not diverge. We use the unequal-cell formulas to account for reference weights and arm imbalance when choosing participant counts and sites. For the specified corrected estimator, we base allocation on variation in corrected squared prediction errors, pilot error, and costs. In simulations, we evaluate estimation error against known prediction risk; in IST, we reconstruct a complete-information corrected estimate in the observed frame. We assess Honduras working intervals under explicit response distributions at recorded village sizes. In both applications, we describe how treatment effects vary with existing site characteristics $W$.

To apply these results, we require valid participant reference distributions, nonselective outcome availability, independent units for constructing moment features, and compatible intervention versions. Effects defined for each site's participant population include both participant composition and site context. For a prospective trial, investigators can use the joint design with recruitment costs and feasible participant counts.

\section*{Funding}
This work was supported by JSPS KAKENHI [grant numbers 22H03552 and 25K21651] and JST CREST [grant number JPMJCR2665].

\section*{Disclosure Statement}
The author declares no competing interests.

\section*{Declaration of Generative AI Use}
The author used ChatGPT6 to assist with checking mathematical proofs, translating text into English, and proofreading the manuscript. Claude Fable 5 was used to assist with writing program code. The author reviewed the resulting text, mathematical arguments, and code and takes responsibility for the final manuscript and reported results.

\section*{Data Availability and Supplementary Materials}
This study reanalyzes existing trial data and involves no new participant recruitment or treatment assignment. The supplementary material below provides complete proofs, numerical protocols, and IST and Honduras measurement and design details.

\bibliographystyle{apalike}
\bibliography{references}

@article{LiLuedtke2023,
  author  = {Li, Sijia and Luedtke, Alex},
  title   = {{Efficient Estimation under Data Fusion}},
  journal = {Biometrika},
  year    = {2023},
  volume  = {110},
  number  = {4},
  pages   = {1041--1054},
  doi     = {10.1093/biomet/asad007},
  note = {DOI: \url{https://doi.org/10.1093/biomet/asad007}}
}

@article{LiGilbertDuanLuedtke2025,
  author  = {Li, Sijia and Gilbert, Peter B. and Duan, Rui and Luedtke, Alex},
  title   = {{Data Fusion Using Weakly Aligned Sources}},
  journal = {Journal of the American Statistical Association},
  year    = {2025},
  volume  = {120},
  number  = {552},
  pages   = {2569--2579},
  doi     = {10.1080/01621459.2025.2476780},
  note = {DOI: \url{https://doi.org/10.1080/01621459.2025.2476780}}
}

@article{Lu2023,
  author  = {Lu, Benjamin and Ben-Michael, Eli and Feller, Avi and Miratrix, Luke},
  title   = {{Is It Who You Are or Where You Are? {A}ccounting for Compositional Differences in Cross-Site Treatment Effect Variation}},
  journal = {Journal of Educational and Behavioral Statistics},
  year    = {2023},
  volume  = {48},
  number  = {4},
  pages   = {420--453},
  doi     = {10.3102/10769986231155427},
  note = {DOI: \url{https://doi.org/10.3102/10769986231155427}}
}

@article{IST1997,
  author  = {{International Stroke Trial Collaborative Group}},
  title   = {{The International Stroke Trial (IST): A Randomised Trial of Aspirin, Subcutaneous Heparin, Both, or Neither among 19,435 Patients with Acute Ischaemic Stroke}},
  journal = {The Lancet},
  year    = {1997},
  volume  = {349},
  number  = {9065},
  pages   = {1569--1581},
  doi     = {10.1016/S0140-6736(97)04011-7},
  note = {DOI: \url{https://doi.org/10.1016/S0140-6736(97)04011-7}}
}

@article{Sandercock2011,
  author  = {Sandercock, Peter A. G. and Niewada, Maciej and Cz{\l}onkowska, Anna and {International Stroke Trial Collaborative Group}},
  title   = {{The International Stroke Trial Database}},
  journal = {Trials},
  year    = {2011},
  volume  = {12},
  pages   = {101},
  doi     = {10.1186/1745-6215-12-101},
  note = {DOI: \url{https://doi.org/10.1186/1745-6215-12-101}}
}

@misc{ISTdata,
  author       = {Sandercock, Peter and Niewada, Maciej and Czlonkowska, Anna},
  title        = {{International Stroke Trial Database (Version 2)}},
  year         = {2011},
  howpublished = {Edinburgh DataShare},
  doi          = {10.7488/ds/104},
  url          = {https://datashare.ed.ac.uk/handle/10283/124},
  note = {DOI: \url{https://doi.org/10.7488/ds/104}. Available at \url{https://datashare.ed.ac.uk/handle/10283/124}}
}

@article{Neyman1934,
  author  = {Neyman, Jerzy},
  title   = {{On the Two Different Aspects of the Representative Method: The Method of Stratified Sampling and the Method of Purposive Selection}},
  journal = {Journal of the Royal Statistical Society},
  year    = {1934},
  volume  = {97},
  number  = {4},
  pages   = {558--625},
  doi = {10.2307/2342192},
  note = {DOI: \url{https://doi.org/10.2307/2342192}}
}

@article{DingFellerMiratrix2019,
  author={Ding, Peng and Feller, Avi and Miratrix, Luke},
  title={{Decomposing Treatment Effect Variation}},
  journal={Journal of the American Statistical Association},
  year={2019}, volume={114}, number={525}, pages={304--317},
  doi={10.1080/01621459.2017.1407322},
  note = {DOI: \url{https://doi.org/10.1080/01621459.2017.1407322}}
}

@article{Hines2025,
  author={Hines, Oliver J. and Diaz-Ordaz, Karla and Vansteelandt, Stijn},
  title={{Variable importance measures for heterogeneous treatment effects}},
  journal={Biometrics}, year={2025}, volume={81}, number={4}, pages={ujaf140},
  doi={10.1093/biomtc/ujaf140},
  note = {DOI: \url{https://doi.org/10.1093/biomtc/ujaf140}}
}

@article{Hoeffding1948,
 author={Hoeffding, Wassily},
 title={{A Class of Statistics with Asymptotically Normal Distribution}},
 journal={Annals of Mathematical Statistics},year={1948},volume={19},number={3},pages={293--325},
 doi={10.1214/aoms/1177730196},
  note = {DOI: \url{https://doi.org/10.1214/aoms/1177730196}}
}

@article{AronowGreenLee2014,
 author={Aronow, Peter M. and Green, Donald P. and Lee, Donald K. K.},
 title={{Sharp Bounds on the Variance in Randomized Experiments}},
 journal={Annals of Statistics},year={2014},volume={42},number={3},pages={850--871},
 doi={10.1214/13-AOS1200},
  note = {DOI: \url{https://doi.org/10.1214/13-AOS1200}}
}

@article{Banerjee2015,
 author = {Banerjee, Abhijit and Duflo, Esther and Goldberg, Nathanael and Karlan, Dean and Osei, Robert and Parient{\'e}, William and Shapiro, Jeremy and Thuysbaert, Bram and Udry, Christopher},
 title = {{A Multifaceted Program Causes Lasting Progress for the Very Poor: Evidence from Six Countries}},
 journal = {Science}, year = {2015}, volume = {348}, number = {6236}, pages = {1260799}, doi = {10.1126/science.1260799},
  note = {DOI: \url{https://doi.org/10.1126/science.1260799}}
}

\clearpage
\setcounter{section}{0}
\setcounter{subsection}{0}
\setcounter{subsubsection}{0}
\setcounter{equation}{0}
\setcounter{table}{0}
\setcounter{figure}{0}
\setcounter{theorem}{0}
\setcounter{proposition}{0}
\setcounter{corollary}{0}
\setcounter{lemma}{0}
\setcounter{assumption}{0}
\setcounter{definition}{0}
\setcounter{remark}{0}
\setcounter{footnote}{0}
\setcounter{page}{1}
\renewcommand{\thepage}{S\arabic{page}}
\setlength{\bibsep}{0pt}
\widowpenalty=10000
\renewcommand{\thesection}{S\arabic{section}}
\renewcommand{\theequation}{S.\arabic{equation}}
\renewcommand{\thetable}{S\arabic{table}}
\renewcommand{\thefigure}{S\arabic{figure}}
\renewcommand{\thetheorem}{S\arabic{theorem}}
\renewcommand{\theproposition}{S\arabic{proposition}}
\renewcommand{\thecorollary}{S\arabic{corollary}}
\renewcommand{\thelemma}{S\arabic{lemma}}

\section*{Supplementary Material}

This supplement provides proofs and implementation details for acquiring and measuring same-site cross-arm information. It contains feature and reference conditions, identification and information results, allocation proofs, numerical protocols, and supporting details for the IST and Honduras applications. References without an S prefix point to the main article.
\section{Proofs of decomposition and identification}
\label{sec:proofs}
For $r_a=\theta_a^*-m_a(W)$, write $\tau=m(W)+r_1-r_0$. Because the residual has conditional mean zero, it is orthogonal to the centered conditional mean. We therefore obtain
\[
 \Var(\tau)=\Var(m)+\E(r_1-r_0)^2
 =\Var(m)+\E\{v_1+v_0-2c\}.
\]
Likewise $\tau-f=(\tau-m)+(m-f)$ has orthogonal terms, proving the risk decomposition. Subtracting the two prediction risks gives the risk difference in the main text. We use the identification result to determine which experimental records identify each term in these population variance and prediction identities.

For the fixed span, the normal equations are $\E\{\bm{b}(\tau-\bm{b}^\trans\bm{\beta})\}=0$. The intercept implies $\E f_b=\mu$, and $\E f_b^2=\bm{t}^\trans \bm{G}^{-1}\bm{t}$. Thus $B_b=\Var(f_b)$ and $R_b=\E(\tau-f_b)^2$. Since $\E(\tau\mid W)=m$, the same normal equations project $m$ onto the span of $\bm{b}$. Orthogonality gives
\[
 \E(m-f_b)^2=\Var(m)-\Var(f_b)=R_b-R.
\]
These identities follow from the normal equations and orthogonality and do not require residual normality.

\begin{proof}[Proof of Theorem~1]
Under randomized arm observation, $\E\{I_aS_a/p_a(W)\mid W\}=m_a(W)$. This identifies $m$, $\mu$, $\bm{t}$, and therefore the explained components. The displayed risk-difference identity proves part (i). Similarly, $\E\{I_aU_{aa}/p_a(W)\mid W\}=\E(\theta_a^{*2}\mid W)$ identifies each conditional arm variance. The only remaining second-order term is $\E(\theta_0^*\theta_1^*\mid W)$. At constant $W$, the two Bernoulli couplings in the main text preserve both entire marginal distributions and change this cross moment. Hence the residual and absolute-risk functionals fail to be identified in the unrestricted model. Finally, $\E\{CU_{01}/\rho(W)\mid W\}=\E(\theta_0^*\theta_1^*\mid W)$ whenever the paired feature is unbiased and $\rho(W)>0$. Substitution proves part (iii).

For part (iv), we write the standardized site vector without its superscript star. Fix $W=w$. A pattern $d$ reveals $\bm{M}_d\bm{\Theta}$ and occurs with positive probability. Let $\bm{\mu}_\Theta=\E(\bm{\Theta}\mid w)$, $\bm{S}_\Theta=\E(\bm{\Theta}\bm{\Theta}^\trans\mid w)$.

For the quadratic result, sufficiency follows from $\tr(\bm{A}\bm{S}_\Theta)=\sum_d\tr\{\bm{B}_d\E(\bm{M}_d\bm{\Theta}\bm{\Theta}^\trans \bm{M}_d^\trans)\}$. If $\bm{A}$ lies outside the stated finite-dimensional span, there is a symmetric matrix $\bm{H}$ orthogonal to that span with $\tr(\bm{A}\bm{H})\ne0$. Orthogonality implies $\bm{M}_d\bm{H}\bm{M}_d^\trans=0$ for every $d$. Two centered Gaussian laws with covariances $\bm{I}$ and $\bm{I}+\epsilon \bm{H}$, for sufficiently small nonzero $\epsilon$, have identical projected distributions and different target values. This proves necessity over the unrestricted model.

For bounded coordinate support, rescale an interior rectangle to $[0,1]^K$. Suppose an off-diagonal pair $(k,\ell)$ is missing. Perturb its uniform product density by $1+\epsilon h(x_k)h(x_\ell)$, with $h(x)=x-1/2$ and sufficiently small $\epsilon$. Every observed coordinate subset omits at least one member of the pair, so integration over that member removes the perturbation. Only the $(k,\ell)$ second moment changes among the moments contributing to the target: all other relevant integrals contain an unpaired mean-zero factor. Its change is nonzero because $\int xh(x)\,dx\ne0$. If an unobserved diagonal $(k,k)$ contributes, perturb the $k$th marginal by a bounded mean-zero function orthogonal to $x$ but not to $x^2$. No observed subset contains $k$. This proves necessity. Sufficiency follows by reading the required pair moment from an observed subset.

The proof identifies conditional moments up to $Q$-null sets; estimation at a continuous $w$ additionally requires appropriate regularity or smoothing.
\end{proof}

\paragraph{Conditional bounds without pairing.}
Conditional Cauchy--Schwarz yields
\[
 \E(\sqrt{v_1}-\sqrt{v_0})^2\le R\le
 \E(\sqrt{v_1}+\sqrt{v_0})^2.
\]
Since $\E\sqrt{v_0v_1}\le\sqrt{\E v_0\E v_1}$, the conditional bounds lie within the corresponding bounds formed from aggregate residual arm variances. We use these inequalities as outer bounds. Researchers can compute sharp bounds under specified marginal laws and support restrictions by solving the corresponding coupling problem, which also determines whether a Cauchy--Schwarz endpoint is attainable.

\section{Moment features and estimation error}
\label{sec:featuredetails}
\begin{proof}[Proof of Proposition~1]
Condition on the site and independent training data. Independence of evaluation units implies $\E(S_{kb}S_{\ell b'})=\theta_k^*\theta_\ell^*$ when $b\ne b'$. Each symmetrized product has this expectation, so their average does as well. This expectation includes sampling independent lower-level units from the site's response population. Randomization alone in a fixed finite list of individual potential outcomes does not make an individual cross-potential-outcome product observable.
\end{proof}

For scalar i.i.d. scores $S_b$ with mean $\theta$ and variance $\sigma^2$, the same-arm statistic has variance
\begin{equation}
 \Var(U)=\frac{4\theta^2\sigma^2}{N}
       +\frac{2\sigma^4}{N(N-1)}.
 \label{eq:Uvariance}
\end{equation}
To verify this, write $S_b=\theta+e_b$ and expand
$U-\theta^2=2\theta\bar e+\{N(N-1)\}^{-1}\sum_{b\ne b'}e_be_{b'}$.
The two terms are orthogonal; only identical or reversed pairs contribute to the second variance. We compute this variance from distinct-score products using finite second moments of the individual scores. Fourth moments of the latent means generally enter the variance across sites of a quadratic feature.

\paragraph{Independent cells with unequal participant counts.}
Within a paired site, treatment arms and participant-reference strata
can have different sample sizes. We derive the covariance of the
corrected moment features using products of distinct observations
within and across cells, retaining the fixed reference weights and
the observed cell counts.

\begin{proposition}[Corrected moments from independent cells]
\label{prop:independent_cell_covariance}
Condition on the site, independent training data, fixed reference
weights, and cell counts. In cell $h=1,\ldots,K$, let
$Y_{h1},\ldots,Y_{hn_h}$ be independent and identically distributed
with mean $\mu_h$ and finite variance $\sigma_h^2$, where $n_h\ge2$.
Assume independence across cells. For fixed signed weights $w_h$,
define
\[
 \tau=\sum_h w_h\mu_h,\qquad
 Z=\sum_h w_h\bar Y_h,\qquad
 \widehat v=\sum_h\frac{w_h^2s_h^2}{n_h},\qquad
 Q=Z^2-\widehat v,
\]
where $s_h^2$ is the sample variance with denominator $n_h-1$.
Write
\begin{align}
 v(\mathbf n)&=\sum_h\frac{w_h^2\sigma_h^2}{n_h},\nonumber\\
 d(\mathbf n)&=
 2\sum_h\frac{w_h^4\sigma_h^4}{n_h(n_h-1)}
 +4\sum_{h<k}\frac{w_h^2w_k^2\sigma_h^2\sigma_k^2}{n_hn_k}.
 \label{eq:cell_variance_terms}
\end{align}
All expectations and covariances below condition on the information
specified above. Then $Z$ and $Q$ are square integrable and satisfy
\begin{align}
 \E Z&=\tau,&\E Q&=\tau^2,\nonumber\\
 \Var(Z)&=v(\mathbf n),&
 \Cov(Z,Q)&=2\tau v(\mathbf n),&
 \Var(Q)&=4\tau^2v(\mathbf n)+d(\mathbf n).
 \label{eq:cell_moment_covariance}
\end{align}
For a predictor value $f$ fixed by the conditioning information,
$L_f=Q-2fZ+f^2$ is square integrable and obeys
\begin{equation}
 \E L_f=(\tau-f)^2,\qquad
 \Var(L_f)=4(\tau-f)^2v(\mathbf n)+d(\mathbf n).
 \label{eq:cell_risk_variance}
\end{equation}
Here $\sigma_h^4=(\sigma_h^2)^2$. These identities require finite
second response moments; Gaussian responses and finite fourth
response moments are unnecessary.
\end{proposition}

\begin{proof}
Let $\varepsilon_{hi}=Y_{hi}-\mu_h$ and
$E_0=\sum_h w_h\bar\varepsilon_h$. Within cell $h$, define
\[
 D_h=\frac1{n_h(n_h-1)}
       \sum_{i\ne i'}\varepsilon_{hi}\varepsilon_{hi'}.
\]
The algebraic identity
$\bar\varepsilon_h^2-s_h^2/n_h=D_h$ gives
\[
 Q=\tau^2+2\tau E_0+D_0,\qquad
 D_0=\sum_h w_h^2D_h
       +2\sum_{h<k}w_hw_k\bar\varepsilon_h\bar\varepsilon_k.
\]
Each summand is a product of distinct independent centered
observations. Its second moment factors into finite response
variances; for example,
$\E(\varepsilon_{hi}^2\varepsilon_{hi'}^2)=\sigma_h^4$ for
$i\ne i'$. Hence $D_0$, $Q$, and $L_f$ are square integrable even
when $Z^2$ and $\widehat v$ separately lack finite second moments.

Every summand in $D_0$ has expectation zero. In a product of a
summand from $E_0$ and a summand from $D_0$, at least one distinct
centered observation appears only once. Conditioning on all other
observations gives expectation zero. The products are integrable
by Cauchy--Schwarz, so $\Cov(E_0,D_0)=0$. The same conditioning
argument makes two different unordered index-pair products
uncorrelated. In the variance of $D_h$, only identical or reversed
ordered pairs remain, giving
\[
 \Var(D_h)=\frac{2\sigma_h^4}{n_h(n_h-1)}.
\]
Across cells, independence gives
\[
 \Var(\bar\varepsilon_h\bar\varepsilon_k)
 =\frac{\sigma_h^2\sigma_k^2}{n_hn_k},\qquad h\ne k.
\]
The within-cell terms and different cross-cell terms in $D_0$ are
mutually uncorrelated. Therefore
$\Var(E_0)=v(\mathbf n)$ and $\Var(D_0)=d(\mathbf n)$.
Since $Z=\tau+E_0$, the decomposition of $Q$ yields
\eqref{eq:cell_moment_covariance}. Finally,
\[
 L_f=(\tau-f)^2+2(\tau-f)E_0+D_0,
\]
which proves \eqref{eq:cell_risk_variance}.
\end{proof}

For the Honduras native-composition contrast, we use the two arm
cells in control--treatment order with weights $(-1,1)$.
By applying the proposition, we obtain the village features $Z_j,Q_j$ and
their covariance for the observed unequal arm counts.
For the IST benefit contrast, we order the four cells as control
age at most 80, control age above 80, treatment age at most 80,
and treatment age above 80. The fixed weights
$(.75,.25,-.75,-.25)$ give the standardized risk reduction.
In both applications, we assume response-sampling independence
conditional on the site and cell counts. Treatment randomization alone in a fixed cohort does not imply
these response-sampling conditions.

\subsection{Unbiased covariance estimation from independent cell products}
\label{sec:cellcovarianceinference}
Using the independent-cell identities in Proposition~\ref{prop:independent_cell_covariance}, we construct a covariance estimator from products of distinct response units. With at least four observations in each cell, this covariance estimator permits inference with fixed lower-level counts as the number of independent sites grows, under the covariance and central limit conditions stated below.

For cell $h$ in site $j$, write $\mu_{hj}=\E(Y_{hij})$ and $m_{2,hj}=\E(Y_{hij}^2)$. For nonnegative integers $r,s$ with $r+s\le4$, define
\begin{equation}
 U_{hj;r,s}=\frac{1}{(n_{hj})_{r+s}}
 \sum_{i_1,\ldots,i_{r+s}\ {\rm distinct}}
 \prod_{a=1}^{r}Y_{hi_a j}
 \prod_{a=r+1}^{r+s}Y_{hi_a j}^{2},\qquad U_{hj;0,0}=1,
 \label{eq:cell_raw_product}
\end{equation}
where $(n)_k=n(n-1)\cdots(n-k+1)$. Every polynomial used below has $r+2s\le4$ in each cell. We expand each such polynomial $p$ in the raw moments $(\mu_{hj},m_{2,hj})$, and define the linear operator $\mathcal U_j$ by replacing each monomial $\prod_h\mu_{hj}^{r_h}m_{2,hj}^{s_h}$ with $\prod_hU_{hj;r_h,s_h}$. In particular, substitute $\sigma_{hj}^2=m_{2,hj}-\mu_{hj}^2$ before expanding the expressions for $\tau_j$, $v_j$, and $d_j$ in Proposition~\ref{prop:independent_cell_covariance}.

\begin{proposition}[Unbiased cell covariance estimator]
\label{prop:unbiased_cell_covariance}
Under the independent-cell model of Proposition~\ref{prop:independent_cell_covariance}, suppose $n_{hj}\ge4$ and $\E(Y_{hij}^2)<\infty$ in every cell. Conditional on the site, counts, fixed weights, and independent training, the matrix
\begin{equation}
 \bm{\widehat K}_j=
 \begin{pmatrix}
  \mathcal U_j(v_j)&\mathcal U_j(2\tau_jv_j)\\
  \mathcal U_j(2\tau_jv_j)&\mathcal U_j(4\tau_j^2v_j+d_j)
 \end{pmatrix}
 \label{eq:unbiased_cell_covariance}
\end{equation}
has integrable entries and satisfies $\E(\bm{\widehat K}_j)=\Cov\{(Z_j,Q_j)^\trans\}$. The same operator gives the unbiased estimator $\mathcal U_j\{4(\tau_j-f_j)^2v_j+d_j\}$ of $\Var(L_{f_j,j})$ for a fixed independently trained predictor $f_j$.
\end{proposition}

\begin{proof}
Distinct indices within a cell are independent, and different cells are independent. Thus each summand in \eqref{eq:cell_raw_product} has expectation $\mu_{hj}^{r}m_{2,hj}^{s}$ and finite absolute expectation. Products across cells have the corresponding product expectation. Linearity gives $\E\{\mathcal U_j(p)\}=p$ for every polynomial in the assertion. Expanding $\tau_j^2v_j$ and $d_j$ requires at most four distinct indices in any cell; for example, $(\sigma_{hj}^2)^2=m_{2,hj}^2-2\mu_{hj}^2m_{2,hj}+\mu_{hj}^4$. The covariance and risk-variance identities in Proposition~\ref{prop:independent_cell_covariance} then prove the claims.
\end{proof}

Unbiasedness concerns the entries of the matrix; individual realized matrices can be indefinite. For example, one cell with weight one, count four, and observed responses $(0,0,1,1)$ gives the entries $(\widehat K_{11},\widehat K_{12},\widehat K_{22})=(1/12,1/12,1/36)$ and a negative determinant. The finite-second-moment condition guarantees integrability, whereas covariance consistency requires an additional law of large numbers.

\paragraph{Conditional finite-frame inference.}
Condition on $J$ sites and an outcome-independent matrix $\bm{H}_J\in\mathbb R^{J\times r}$ with orthonormal columns, fixed $r$, and first column $\one/\sqrt J$. Write
\[
 \bm{\beta}_J=(\bm{\gamma}_J^\trans,q_J)^\trans,
 \quad \bm{\gamma}_J=\bm{H}_J^\trans\bm{\tau}/\sqrt J,
 \quad q_J=\bm{\tau}^\trans\bm{\tau}/J,
 \quad \bm{\widehat\beta}_J=(\bm{Z}^\trans \bm{H}_J/\sqrt J,\bar Q)^\trans.
\]
Let $\bm{h}_j$ be the transpose of row $j$ of $\bm{H}_J$, and define
\[
 \bm{D}_j=\begin{pmatrix}\bm{h}_j&\bm{0}_{r\times1}\\0&J^{-1/2}\end{pmatrix},\quad
 \bm{\xi}_j=\bm{D}_j\begin{pmatrix}Z_j-\tau_j\\Q_j-\tau_j^2\end{pmatrix},\quad
 \bm{\Omega}_J=\sum_j\bm{D}_j\bm{K}_j\bm{D}_j^\trans,\quad
 \bm{\widehat\Omega}_J=\sum_j\bm{D}_j\bm{\widehat K}_j\bm{D}_j^\trans.
\]
Here $\bm{K}_j=\Cov\{(Z_j,Q_j)^\trans\}$ and $\sqrt J(\bm{\widehat\beta}_J-\bm{\beta}_J)=\sum_j\bm{\xi}_j$. By Proposition~\ref{prop:unbiased_cell_covariance}, we obtain $\bm{\widehat\Omega}_J$ as an unbiased estimator of this covariance.

\begin{corollary}[Conditional projection coverage]
\label{cor:finiteframe_projectioncoverage}
Suppose the site records are independent and $\bm{\Omega}_J\succ0$. For every $\epsilon>0$, assume the whitened Lindeberg condition
\[
 \sum_j\E\!\left[\|\bm{\Omega}_J^{-1/2}\bm{\xi}_j\|^2
 \one\{\|\bm{\Omega}_J^{-1/2}\bm{\xi}_j\|>\epsilon\}\right]\longrightarrow0
\]
and the relative covariance consistency condition
\[
 \|\bm{\Omega}_J^{-1/2}(\bm{\widehat\Omega}_J-\bm{\Omega}_J)
       \bm{\Omega}_J^{-1/2}\|_{\rm op}\longrightarrow_p0.
\]
When $\bm{\widehat\Omega}_J\succ0$, form the $(1-\alpha)$ confidence ellipsoid for $\bm{\beta}$ using
\[
 J(\bm{\widehat\beta}_J-\bm{\beta})^\trans\bm{\widehat\Omega}_J^{-1}
        (\bm{\widehat\beta}_J-\bm{\beta})\le\chi^2_{r+1,1-\alpha},
\]
intersect it with $q\ge\|\bm{\gamma}\|^2$, and project onto
$B_F=\sum_{k=2}^r\gamma_k^2$, $R_F=q-\|\bm{\gamma}\|^2$, and $V_F=q-\gamma_1^2$. If the estimated covariance is not positive definite or the intersection is empty, report the full feasible component ranges. The projected joint region has limiting coverage at least $1-\alpha$, including when a component is zero.
\end{corollary}

\begin{proof}
The multivariate Lindeberg theorem gives $\bm{\Omega}_J^{-1/2}\sum_j\bm{\xi}_j\rightsquigarrow N(0,\bm{I}_{r+1})$. Relative covariance consistency implies that the estimated matrix is positive definite with probability tending to one and that the studentized quadratic form has a $\chi^2_{r+1}$ limit. The true primitive vector belongs to $q\ge\|\bm{\gamma}\|^2$ because $\bm{H}_J$ has orthonormal columns. Whenever the primitive ellipsoid contains that vector, its feasible intersection and its component projection contain the true targets. Returning the full feasible ranges on the specified exceptional events preserves this inclusion.
\end{proof}

One sufficient setting allows fixed cell counts: the number of cells and their absolute weights are uniformly bounded, $n_{hj}\ge4$, the response fourth moments are uniformly bounded, $\max_j\|\bm{h}_j\|\to0$, and the eigenvalues of $\bm{\Omega}_J$ are uniformly bounded above and away from zero. The distinct-index representations give uniformly bounded fourth moments of $Z_j,Q_j$ and second moments of the entries of $\bm{\widehat K}_j$. Thus $\sum_j\E\|\bm{\xi}_j\|^4$ is bounded by a constant times $\sum_j\|\bm{h}_j\|^4+J^{-1}\to0$, because $\sum_j\|\bm{h}_j\|^4\le r\max_j\|\bm{h}_j\|^2$. This proves the Lindeberg condition. Independence across sites bounds the variances of entries in the three blocks of $\bm{\widehat\Omega}_J-\bm{\Omega}_J$ by constant multiples of $\sum_j\|\bm{h}_j\|^4$, $J^{-1}\sum_j\|\bm{h}_j\|^2$, and $J^{-1}$, respectively, establishing relative covariance consistency. Under these conditions, we obtain asymptotic projection coverage for the specified finite-frame targets. For the empirical Honduras intervals, we use the sample-variance plug-in calculation described in Section~\ref{sec:hondurasprotocol} and evaluate its approximation separately.

\begin{proposition}[Mean and covariance errors in products]
Let evaluation scores obey $\E(S_{kb}\mid\mathcal G)=\theta_k+\delta_{kb}$ and let $\Gamma_{kb,\ell b'}=\Cov(S_{kb},S_{\ell b'}\mid\mathcal G)$. Then the bias of a cross-unit product is
\[
 \E(S_{kb}S_{\ell b'}\mid\mathcal G)-\theta_k\theta_\ell
 =\theta_k\delta_{\ell b'}+\theta_\ell\delta_{kb}
  +\delta_{kb}\delta_{\ell b'}+\Gamma_{kb,\ell b'}.
\]
For weighted averages of symmetrized pairs, average these four terms with the same weights as the estimator.
\end{proposition}
\begin{proof}
Use $\E(XY)=\E X\E Y+\Cov(X,Y)$ and expand the product of the two conditional means.
\end{proof}
We exclude both evaluation units from nuisance estimation to prevent either outcome from entering the estimated function used in its own product. We account separately for density-ratio estimation error, unequal conditional score means, and uncertainty in a shared estimated reference through the product-bias identity.

\paragraph{Standardized mean scores.}
With fixed $P^*$, density ratio $r_j=dP^*/dP_j$, treatment propensity $e_a$, and independent training estimates $(\widehat r_j,\widehat m_{aj})$, we construct the standardized mean score
\[
 S_{aij}=P^*\widehat m_{aj}
  +\frac{\one(A_i=a)\widehat r_j(\bm{X}_i)}{e_a(\bm{X}_i)}
        \{Y_i-\widehat m_{aj}(\bm{X}_i)\}.
\]
Conditional on training and the site,
\[
 \E S_{aij}-\theta_{aj}^*
 =P_j\{(\widehat r_j-r_j)(m_{aj}-\widehat m_{aj})\}.
\]
Cauchy--Schwarz bounds this bias by the product of the two nuisance $L_2(P_j)$ errors. We can include independent missingness weights under the corresponding conditional response assumption and apply the product-bias identity to products of these scores. With fixed small $n_j$, site-specific nuisance estimates need not become accurate as $J$ grows; an exact or negligible-error argument then requires known-weight constructions or justified pooling across sites.

\paragraph{Shared references and covariance.}
For a fixed symmetric residual-maker $\bm{M}$, condition on the latent site means and covariate matrix, and write $\bm{\widehat\theta}_a=\bm{\theta}_a+\bm{\varepsilon}_a$ with mean-zero errors. For $\bm{\Omega}_{ab}=\E(\bm{\varepsilon}_a\bm{\varepsilon}_b^\trans)$,
\[
 \E(\bm{\widehat\theta}_a^\trans \bm{M}\bm{\widehat\theta}_b)
 =\bm{\theta}_a^\trans \bm{M}\bm{\theta}_b+\tr(\bm{M}\bm{\Omega}_{ba}).
\]
We obtain this expectation by expanding the bilinear form and applying the trace identity, including cross-site and cross-arm covariance from a common fit. If $\bm{M}$ is learned from outcomes, conditioning must also justify its independence or account for the learning error.

If an external reference sample of size $n_*$ supplies an additional expansion $n_*^{-1}\sum_r\bm{\zeta}(\bm{X}_r^*)$, that contribution is added to the site expansion. Independence gives covariance $\bm{V}_{\mathrm{site}}/J+\bm{V}_*/n_*$, where $\bm{V}_{\mathrm{site}}$ and $\bm{V}_*$ are the covariance matrices of the site and reference influence contributions. An internal reference generally shares observations with the site scores and requires their covariance as well. The IST analysis uses the fixed weights $(.75,.25)$ and therefore has no estimated common-reference term.

\paragraph{Transfer conditions used in the main theorem.}
Let $\widehat Z_k=Z_k+e_k$ denote actual and oracle signal records. Sufficient first-order conditions are
\[
 \sqrt J\,\mathbb P_Je_1=o_p(1),\quad
 \sqrt J\,\mathbb P_J(\widehat f e_1)=o_p(1),\quad
 \sqrt{Jr_J}\,\mathbb P_Je_2=o_p(1).
\]
For example a quadratic proxy remainder $r_{2j}$ contributes
$\sqrt{r_J/J}\sum_j C_jr_{2j}/\rho_J(W_j)$ to the last condition. Nuisance errors of augmentation predictors under known probabilities are centered conditional on independent training; their $L_2$ consistency controls their rate-scaled variance. For feasible covariance estimation also require
$J^{-1}\sum_j\|\bm{\widehat\xi}_j-\bm{\xi}_j\|^2=o_p(1)$, where $\bm{\xi}_j$ is the vector of rate-scaled influence contributions. An aggregate mean condition alone cannot imply this empirical squared-error condition.

\section{Proofs of joint inference}
\label{sec:proofrates}
\begin{proof}[Proof of Theorem~3]
First consider the oracle signals. On an evaluation fold condition on its training sample and put $e=\widehat f-m$. The change in the explained score is
\[
 (2\widehat fZ_1-\widehat f^2)-(2mZ_1-m^2)
 =2e(Z_1-m)-e^2.
\]
Its conditional mean is $-\|e\|_2^2$. The conditional variance of the first term is bounded by a constant times $\|e\|_2^2$, and the second term has second moment at most $\|e\|_4^4$. Hence the centered evaluation average is $o_p(J^{-1/2})$. The conditional mean is negligible under the stated squared-error condition. The residual score changes by the negative of this quantity, for which the weaker $\sqrt{Jr_J}$ rate requires only the weaker condition stated for $R$ alone. There are finitely many folds, so adding these bounds proves the expansion for all sites without asserting independence between training folds.

Since $\widehat\mu=\mathbb P_JZ_1$,
\[
 \widehat\mu^2-\mu^2=2\mu\mathbb P_J(Z_1-\mu)+O_p(J^{-1}).
\]
Combining the preceding expansions gives the displayed influence functions. Their expectations are zero. The multivariate triangular-array Lindeberg theorem gives the limit with covariance $\bm{\Sigma}$.

For the fixed-span case let $e=\bm{b}^\trans(\bm{\widehat\beta}-\bm{\beta})$. Its extra conditional mean term $2\E\{e(m-f_b)\}$ is zero by the population normal equations. Root-$J$ coefficient consistency gives $\E(e^2)=O_p(J^{-1})$. The moment conditions on $\bm{b}$ and $Z_1-f_b$ bound the corresponding centered evaluation terms by $o_p(J^{-1/2})$. This proves the expansion with $f_b$ without assuming $m=f_b$.

The stated aggregate conditions absorb differences between estimated and oracle signals. To prove covariance consistency, let $\bm{\Delta}_j=\bm{\widehat\xi}_j-\bm{\xi}_j$. Each cross term in the difference of empirical covariance matrices is bounded by
\[
 \left(J^{-1}\sum_j\|\bm{\Delta}_j\|^2\right)^{1/2}
 \left(J^{-1}\sum_j\|\bm{\xi}_j\|^2\right)^{1/2}=o_p(1).
\]
The quadratic error is also negligible, and the oracle second-moment law of large numbers supplies the limit. Centering the influence vectors changes neither limit. Slutsky's theorem gives studentized inference for nondegenerate contrasts.
\end{proof}

For regular i.i.d. records with fixed observation probabilities, the bounded conditional variance and $L_4$ assumptions are convenient sufficient conditions. In a rare-pair array it is the rate-scaled influences that need a triangular-array moment condition. For bounded paired features, $\,Jr_J\to\infty$ and probabilities proportional to $r_J$ imply a Lyapunov condition of any fixed order above two. The lower-level feature bounds must themselves remain controlled; increasing inverse standardization weights can defeat that conclusion.

\paragraph{Information bound for residual variation.}
\label{sec:lowerboundproof}
Take constant $W$ and an oracle Gaussian site vector with known marginal means, positive standard deviations $\sigma_0,\sigma_1$, and unknown correlation $\gamma\in(-1,1)$. Single-arm records have laws that do not depend on $\gamma$. Each paired record has Fisher information
\[
 I_1(\gamma)=\frac{1+\gamma^2}{(1-\gamma^2)^2}.
\]
For independent Bernoulli selection of paired sites with probability $r_J$, total information is $Jr_JI_1(\gamma)$. On a compact interior neighborhood, the Gaussian log likelihood has bounded score moments and differentiable quadratic-mean structure. If $Jr_J\to\infty$, its local expansion at $\gamma_0+h/\sqrt{Jr_J}$ is
\[
 h\Delta_J-\tfrac12h^2I_1(\gamma_0)+o_p(1),\qquad
 \Delta_J\rightsquigarrow N\{0,I_1(\gamma_0)\}.
\]
To verify the expansion, only paired records enter. Their score sum has variance $Jr_JI_1$; Lyapunov's bound is of order $(Jr_J)^{-\delta/2}$. The normalized Hessian converges by the same argument, and the third-order remainder is $O_p\{(Jr_J)^{-1/2}\}$ for bounded local $h$.

The local asymptotic minimax risk for $\gamma$ on the $Jr_J$ scale is at least $I_1(\gamma_0)^{-1}$. Since
\[
 R=V=\sigma_0^2+\sigma_1^2-2\gamma\sigma_0\sigma_1,
\]
the corresponding variance-target bound is $4\sigma_0^2\sigma_1^2/I_1(\gamma_0)$. If $Jr_J$ remains bounded, two nearby correlations have bounded Kullback--Leibler distance. Choosing their separation small enough gives a testing error bounded away from zero; an estimator uniformly consistent on a fixed nondegenerate neighborhood would contradict this two-point bound. This proves the rate restriction for a model containing the submodel. This information bound applies at interior correlations; we handle a zero-variance boundary separately by projection.

\begin{proof}[Proof of Proposition~2]
Within stratum $h$, $\E\bar Z_{1h}^2=m_h^2+\Var(Z_1\mid h)/J_h$ and $s_{1h}^2$ is unbiased for the conditional variance. Thus $\E M_{2h}=m_h^2$. Independent stratum samples give
$\Var(\widehat\mu)=\sum_hq_h^2\Var(Z_1\mid h)/J_h$, so $\E M_{2\mu}=\mu^2$. The second-moment signal has expectation $q_{2h}$. Substitution proves all three unbiasedness statements. Their recordwise equality follows by cancellation.
\end{proof}

We can also write $M_{2h}$ as the average $Z_{1j}Z_{1k}$ over distinct sites in stratum $h$. Independence between these sites gives the stated unbiasedness. When a shared outcome-trained function induces dependence between sites, this argument requires conditioning on an independent pilot. For i.i.d. sites and fixed $H$, the order-$J^{-1}$ corrections do not alter a nondegenerate first-order expansion.

\section{Confidence sets at the boundary}
\label{sec:projectioninference}
Let $[l_h,u_h]$ and $[L_h,U_h]$ be simultaneous confidence intervals for $m_h$ and $q_{2h}$. For the application calculation, we use the critical value $z_{1-\alpha/(4H)}$ for the $2H$ two-sided intervals. Intersect them with $m_h\in[-1,1]$ and $m_h^2\le q_{2h}\le1$. If $U_h<0$, $L_h>U_h$, or the clipped mean interval is empty in any stratum, record an empty region and return the full feasible range $[0,1]$ for the three variance components.

Otherwise, we replace $L_h$ by $\max(0,L_h)$, $U_h$ by $\min(1,U_h)$, and intersect $[l_h,u_h]$ with $[-\sqrt{U_h},\sqrt{U_h}]$. Let
\[
 a_h=\min_{x\in[l_h,u_h]}x^2,\qquad d_h=\max(l_h^2,u_h^2).
\]
Then valid projections for the conditional residual variance are
\[
 v_h\in[\max(0,L_h-d_h),\ U_h-a_h].
\]
Because the rectangle separates strata, summing these endpoints with $q_h$ gives the extrema for $R_b$. To obtain the minimum explained variance, which depends only on the mean vector, we solve the convex program
\[
 \min_{x_h\in[l_h,u_h]}\left\{\sum_hq_hx_h^2-(\sum_hq_hx_h)^2\right\},
\]
and the maximum explained variance is attained at a vertex of the mean rectangle. With three strata, we enumerate all eight vertices. We report a valid interval for $V$ as $[0,V_U]$, with
\[
 V_U=\sum_hq_hU_h-\min_{z\in[\sum_hq_hl_h,\sum_hq_hu_h]}z^2.
\]
The upper endpoint is attained within this primitive rectangle by taking all $q_{2h}=U_h$ and a mean vector achieving the displayed weighted mean. The lower endpoint zero is conservative when the full region excludes zero total variance. We report zero as a conservative lower bound rather than the exact lower projection.

\begin{proposition}[Coverage transfer]
If the primitive rectangle covers all $(m_h,q_{2h})$ with limiting probability at least $1-\alpha$, the projected intervals cover $(B_b,R_b,V)$ and every reported $v_h$ jointly with limiting probability at least $1-\alpha$.
\end{proposition}
\begin{proof}
On the event that the primitive rectangle covers the true vector, it has a nonempty intersection with the feasible parameter space, since the true vector is feasible. Every functional of that vector lies between the extrema or valid bounds over the intersection. Returning the full feasible range when the intersection is empty cannot remove a coverage event. Taking probabilities proves the assertion.
\end{proof}
Normal intervals for the primitive moments require sufficient sites and a nondegenerate variance estimate for the relevant moments. We assess their accuracy separately from point-estimator unbiasedness by reporting coverage of the projected intervals, their widths, the frequency of empty intersections, and ordinary Wald coverage.

\section{Identification under general observation patterns}
For a quadratic form, only the symmetric part of the target matrix contributes. A cross-term $c\theta_a\theta_b$, with $a\ne b$, is represented by symmetric entries $A_{ab}=A_{ba}=c/2$.
\label{sec:generalpatterns}
Fix $W=w$. A pattern $d$ reveals $\bm{M}_d\bm{\Theta}$ and occurs with positive probability. Let $\bm{\mu}_\Theta=\E(\bm{\Theta}\mid w)$, $\bm{S}_\Theta=\E(\bm{\Theta}\bm{\Theta}^\trans\mid w)$. We characterize which linear and quadratic moment targets are identified by these observation patterns.

\begin{theorem}[Linear and quadratic pattern characterization]
\label{thm:generalid}
If the distribution of each observed $\bm{M}_d\bm{\Theta}$ is known, then $\bm{a}^\trans\bm{\mu}_\Theta$ is identified over the unrestricted finite-first-moment model if and only if $\bm{a}$ is in the span of the rows of the $\bm{M}_d$. For symmetric $\bm{A}$, $\tr(\bm{A}\bm{S}_\Theta)$ is identified over the unrestricted finite-second-moment model if and only if
\[
 \bm{A}\in\operatorname{span}\{\bm{M}_d^\trans \bm{B}_d\bm{M}_d:\bm{B}_d=\bm{B}_d^\trans,\ d\text{ available}\}.
\]
For coordinate selectors and a bounded rectangular support with nonempty interior, the latter holds exactly when every nonzero coordinate pair of $\bm{A}$ is jointly observed by some pattern, in a model containing all sufficiently small bounded-density perturbations on that rectangle.
\end{theorem}
\begin{proof}
If $\bm{a}=\sum_d\bm{M}_d^\trans \bm{b}_d$, its mean is $\sum_d\bm{b}_d^\trans\E(\bm{M}_d\bm{\Theta})$. If $\bm{a}$ lies outside that span, select $\bm{v}$ in the common null space of $\bm{M}_d$ with $\bm{a}^\trans \bm{v}\ne0$ and shift any integrable distribution by a small multiple of $\bm{v}$. All observations have unchanged laws, while the target mean changes.

The quadratic characterization, including the bounded-coordinate-support case, follows from the proof of Theorem~1(iv) in Supplement Section~\ref{sec:proofs}.
\end{proof}

For target matrices $\bm{A}_r$ and costs $c_d$, we select the observation patterns needed for identification through the integer covering problem
\[
 \min_{y_d\in\{0,1\}}\sum_dc_dy_d,\qquad
 \sum_{d:k,\ell\in S_d}y_d\ge1
 \quad\text{for every pair }(k,\ell)\text{ required by some }\bm{A}_r.
\]
Having selected the identifying observation patterns, we allocate sites to control estimation precision. With finite within-site samples, we use corrected moment features to identify the corresponding moments of the latent site vector.

\paragraph{Oracle information and generated-feature information.}
In the oracle partial-coordinate model, define $\mathcal Kf=\E(f\mid O)$ and $\mathcal A=\mathcal K^*\mathcal K$ on centered square-integrable functions of $(W,\bm{\Theta})$. For a centered target gradient $h_0$, a finite information bound requires $h_0\in\operatorname{Range}(\mathcal A^{1/2})$. Indeed the observation-space gradient solves $\mathcal K^*g=h_0$; polar decomposition gives $\operatorname{Range}(\mathcal K^*)=\operatorname{Range}(\mathcal A^{1/2})$. Its minimum squared norm is the inverse spectral moment $\int_{(0,1]}\lambda^{-1}\,d\langle E_\lambda h_0,h_0\rangle$, where $E_\lambda$ is the spectral resolution. A null-space component prevents a solution. Membership in the closure of the range is insufficient when the inverse moment diverges. If $h_0\in\operatorname{Range}(\mathcal A)$, the minimum-norm expression is $\mathcal K\mathcal A^\dagger h_0$.

We calculate oracle information from the site means. For the scalar experiment with corrected moment features observing $(W,C,CL_f)$, we use the augmented influence function given in the main text. We include the measurement variation in $L_f$ when calculating its variance to obtain the conditional allocation criterion. These calculations do not establish semiparametric efficiency when all participant records are available.

\section{Allocation proofs and implementation}
\label{sec:designproofs}
\begin{proof}[Proof of Theorem~4]
For $a=a_f$, the residual $C(L_f-a)/\rho$ has conditional mean zero and conditional variance $\nu_f^2/\rho$. Its covariance with $a-\mathcal R(f)$ is zero. The law of total variance proves the first formula.

For a pilot predictor $\widetilde a$, condition on the pilot and write $e=\widetilde a-a$. The estimator's conditional mean given $W$ is $a$, and its residual about that mean is $C(L_f-a)/\rho+(1-C/\rho)e$. The covariance of these two terms is zero because $\E(L_f-a\mid W)=0$. The second has variance $(\rho^{-1}-1)e^2$, proving the penalty identity. The penalty identity specifies the augmentation error relevant to allocation for this fixed estimator.

With the oracle predictor, minimize the convex function $\nu_f^2(w)/r+\lambda c(w)r$ at each $w$ over the prescribed interval. The unconstrained stationary point is the square root in the main text; clipping gives the constrained minimizer. Integrating the pointwise optimality inequalities and using equality of expected cost at a binding optimum proves global optimality. Zero-variance locations use their lower bound. Boundary budgets and a nonbinding budget use the corresponding endpoints. For the fixed pilot predictor the only allocation-dependent part is $\E\{(\nu_f^2+e^2)/\rho\}$, so the same proof applies.
\end{proof}

\begin{proposition}[Pilot error and design regret]
For a finite partition with known masses, fixed bounds $0<\ell\le\rho_h\le u$, and a common feasible budget, put $F_s(\rho)=\sum_hq_hs_h/\rho_h$. If $\max_h|\widehat s_h-s_h|\le\epsilon$, an exact minimizer $\widehat\rho$ for the estimated criterion satisfies
\[
 0\le F_s(\widehat\rho)-F_s(\rho^*)\le2\epsilon/\ell.
\]
An optimization error at most $\delta$ adds $\delta$ to this bound.
\end{proposition}
\begin{proof}
Uniformly over the feasible set, $|F_{\widehat s}-F_s|\le\sum_hq_h\epsilon/\ell=\epsilon/\ell$. Insert and subtract the estimated criterion at $\widehat\rho$ and $\rho^*$; its minimizing inequality bounds the middle difference by zero (or $\delta$). The result follows.
\end{proof}
We use this bound to control the objective regret of a pilot-based allocation. An $m^{-1}$ regret rate additionally requires smooth interior solutions, positive curvature, and root-$m$ pilot estimates. If probability floors tend to zero, we must evaluate pilot accuracy relative to those floors. To stabilize variance estimates for an estimated pilot predictor, we can cross-fit within the pilot or use upper confidence bounds. For the numerical rule, we specify a variance floor and evaluate the resulting allocation on independent sites.

\paragraph{Fixed-count assignment.}
In stratum $h$, let fixed residual records be $e_j=L_{fj}-\widetilde a_f(W_j)$, with a simple random sample of $k_h$ paired sites among $J_h$. The estimator of the frame mean is
\[
 \bar a_h+\frac1{k_h}\sum_{j:C_j=1}e_j.
\]
Its conditional design variance is $(1-k_h/J_h)S_{eh}^2/k_h$, where $S_{eh}^2=(J_h-1)^{-1}\sum(e_j-\bar e_h)^2$. Independent strata give variance $\sum_hq_h^2(1-k_h/J_h)S_{eh}^2/k_h$. If the frame itself is a random sample from a site population, the law of total variance adds the variance of the full-frame target. Thus a finite-population correction belongs to the deletion component, not to all population sampling uncertainty.

For Bernoulli assignment in a fixed evaluation frame, as used in the IST information-deletion experiment,
\[
 \Var_C(\widehat{\mathcal R}_f\mid\text{frame,pilot})
 =J^{-2}\sum_j(\rho_j^{-1}-1)(L_{fj}-\widetilde a_f(W_j))^2.
\]
We compute this conditional variance exactly and also report Monte Carlo deletion RMSE. We compare allocations conditional on the realized frame, using the full-information corrected estimate as the deletion reference rather than as the population risk.

\paragraph{Integer allocation and cost accounting.}
For fixed counts, round desired $J_h\rho_h$ within capacity bounds, enforce at least the required paired count in every identifying stratum, and allocate remaining affordable counts to the greatest decrease in the prespecified variance criterion per cost. With a finite number of strata and probabilities bounded away from zero, changing each count by $O(1)$ changes an average risk by $O(J^{-1})$ on the normalized count scale. Actual random selection within each stratum preserves known inclusion probabilities. Bernoulli assignment enforces an expected cost; a hard budget requires the fixed-count procedure.

We account for the independent pilot cost separately from evaluation-site costs. Both IST allocations use the same pilot, and both simulation designs use the same pilot size. We compare incremental allocation while holding the number of evaluation sites and pilot size fixed.

\paragraph{General information program.}
For information matrices $\bm{\mathcal I}_d\succeq0$ corresponding to a specified finite-dimensional local experiment, let $\bm{\mathcal I}(\bm{n})=\sum_dn_d\bm{\mathcal I}_d$. For target weight $\bm{H}\succeq0$, define
\[
 \mathfrak R(\bm{n})=\begin{cases}
 \tr\{\bm{H}\bm{\mathcal I}(\bm{n})^\dagger\},&\operatorname{Range}(\bm{H})\subseteq\operatorname{Range}\{\bm{\mathcal I}(\bm{n})\},\\
 +\infty,&\text{otherwise}.
 \end{cases}
\]
It is convex and lower semicontinuous because it is the supremum over matrices $\bm{Z}$ of
$2\tr(\bm{H}^{1/2}\bm{Z})-\tr\{\bm{Z}^\trans\bm{\mathcal I}(\bm{n})\bm{Z}\}$. On a nonempty compact budget-capacity set with a finite-risk feasible point, a minimizer exists. On a constant-range face its active interior KKT conditions equate
\[
 \tr\{\bm{H}\bm{\mathcal I}(\bm{n})^\dagger\bm{\mathcal I}_d\bm{\mathcal I}(\bm{n})^\dagger\}=\lambda c_d.
\]
For planning, we use information matrices from the finite-site experiment that produces the corrected moment features; latent, noise-free moments correspond to a different experiment. When both the number of sites and paired fraction vary and risk has the special form
$\mathcal B^{-1}(c_C+\rho\Delta c)(V_M+V_X/\rho)$, differentiation yields
$\rho^*=\operatorname{clip}\{\sqrt{c_CV_X/(\Delta cV_M)},\rho_{\min},1\}$ for positive components. This scalar expression applies only when that decomposition is valid for the selected estimator.

\section{Incomplete pairing support and delivery sensitivity}
\label{sec:extensionproofs}
\subsection{Incomplete pairing support}
\begin{proof}[Proof of Theorem~2]
For each $w$, the conditional marginal second moments and $m(w)$ are fixed by the observed marginals. The conditional variance is therefore affine and decreasing in the cross moment. The rearrangement inequality bounds that cross moment between $K_-(w)$ and $K_+(w)$. The lower variance endpoint is attained by $(\theta_0^*,\theta_1^*)=(q_0(U,w),q_1(U,w))$ and the upper variance endpoint by $(q_0(U,w),q_1(1-U,w))$, where $U$ is conditionally uniform. Generalized quantiles are measurable for probability kernels on the real line, so these are admissible conditional couplings. Use the identified coupling on $\mathcal E$ and either extremal coupling on $\mathcal E^c$. This attains each integrated endpoint without changing an observed marginal or an observed paired distribution. Mixtures of the endpoint kernels attain every intermediate integrated variance because the conditional means and marginal second moments are fixed. This establishes sharpness. The width formula follows by subtraction. Its integrand is nonnegative; an integral of a nonnegative integrable function is zero exactly when the integrand is zero almost surely. Adding an identified constant preserves sharpness. Conditional Cauchy--Schwarz gives the moment-only outer bounds.
\end{proof}

\subsection{Sensitivity to a change in intervention delivery}
\label{sec:versionsensitivity}
Whole-site and within-site assignment may change how an intervention is delivered. Suppose pairing identifies a site effect $\tau^{p}$, whereas the scientific target is $\tau^{c}$ under whole-site delivery. Define $\mu^p,B^p,R^p$ and $\mu^c,B^c,R^c$ by \textup{(2)}, using the same $W$ and site distribution. A quantitative sensitivity assumption is
\begin{equation}
 \|\tau^c-\tau^p\|_2\le\delta.
 \label{eq:versionL2}
\end{equation}
The parameter $\delta$ bounds the root-mean-square delivery difference over sites. It can be stated on the same outcome scale as the site effects, and is zero when the delivery versions agree.

\begin{theorem}[Joint sensitivity of the heterogeneity decomposition]
\label{thm:versiondecomp}
Under \eqref{eq:versionL2}, the target decomposition satisfies
\begin{equation}
 (\mu^c-\mu^p)^2+
 (\sqrt{B^c}-\sqrt{B^p})^2+
 (\sqrt{R^c}-\sqrt{R^p})^2\le\delta^2.
 \label{eq:versionball}
\end{equation}
For a fixed paired-effect distribution on an unrestricted square-integrable probability space, this region is sharp if the mean-zero $W$-measurable subspace and its orthogonal residual subspace are both nontrivial. The same result holds for $(B_b,R_b)$ with the corresponding projection subspaces. With bounded outcomes, \eqref{eq:versionball} remains a valid outer region; sharpness then additionally depends on the outcome support. For any fixed predictor $f$,
\begin{equation}
 \big|\sqrt{\mathcal R_c(f)}-\sqrt{\mathcal R_p(f)}\big|\le\delta.
 \label{eq:versionrisk}
\end{equation}
\end{theorem}

We obtain the joint restriction by decomposing each effect into its mean, centered conditional mean, and residual, which lie in mutually orthogonal subspaces. The squared norm of the delivery difference is shared by changes in the overall mean and the square roots of the two heterogeneity components. We use the shared radius $\delta$ to construct joint sensitivity ranges for all three coordinates. We can intersect the region with any mean-effect information supplied by whole-site assignments. If a confidence region covers the paired decomposition, taking the union of \eqref{eq:versionball} over that region transfers its coverage to the delivery-sensitive target, provided the stated bound $\delta$ is valid. Under the specified bound on delivery differences, researchers can report the resulting ranges for heterogeneity components and prediction risk.

\begin{proof}[Proof of Theorem~\ref{thm:versiondecomp}]
Work in the Hilbert space $L_2(Q)$, enlarged to contain the latent site variables. Let $P_0$ be projection onto constants and $P_W$ projection onto $W$-measurable functions. The three mutually orthogonal projections are $P_0$, $P_W-P_0$, and $I-P_W$. Write $d=\tau^c-\tau^p$. Pythagoras gives
\[
 \|d\|_2^2=\|P_0d\|_2^2+\|(P_W-P_0)d\|_2^2+\|(I-P_W)d\|_2^2.
\]
The first term equals $(\mu^c-\mu^p)^2$. The reverse triangle inequality bounds the other two terms below by the squared differences of the corresponding projected norms. These norms are $\sqrt B$ and $\sqrt R$, proving \eqref{eq:versionball}.

For sufficiency, fix any proposed $(\mu^c,b_c,r_c)$ with $b_c,r_c\ge0$ satisfying the norm inequality, where $b_c=\sqrt{B^c}$ and $r_c=\sqrt{R^c}$. In each of the two nonconstant subspaces take the unit vector in the direction of the paired projection when it is nonzero. When it is zero, take any unit vector in that subspace, whose existence is assumed. Set $\tau^c$ equal to $\mu^c$ plus $b_c$ times the first direction plus $r_c$ times the second. Its decomposition has the required values and its squared distance from $\tau^p$ equals the left side of \eqref{eq:versionball}. This construction requires an unrestricted square-integrable effect space; it may violate a bounded outcome support, which is why only outer validity is claimed there. Replacing $P_W$ by projection onto the span of $\bm{b}$ gives the finite-basis result. The risk statement is the reverse triangle inequality applied to $\tau^c-f$ and $\tau^p-f$.
\end{proof}

\section{Small explained components and an estimated reference}
\subsection{Inference near zero explained variation}
\label{sec:weakexplanationdetails}
At a constant conditional mean, the first derivative of $B$ vanishes, so its estimator has a different rate and limiting distribution. We characterize this second-order behavior using the finite-stratum correction. Suppose $H$ is fixed, $J_h/J\to q_h>0$, the target masses $q_h$ are fixed, and site signals in stratum $h$ have variance $\sigma_h^2>0$. Put
\[
 \bm{A}_q=\diag(\bm{q})-\bm{q}\bm{q}^\trans,\qquad
 \bm{\Omega}=\diag(\sigma_h^2/q_h),\qquad
 \bm{\widehat D}=\diag(s_{1h}^2/J_h).
\]
The estimator of the explained component can be written as
$\widehat B_b=\bm{\bar Z}_1^\trans \bm{A}_q\bm{\bar Z}_1-\tr(\bm{A}_q\bm{\widehat D})$. This expression also applies to prescribed mixture masses differing from the sampling fractions if $\bm{\Omega}$ uses the limiting sampling fractions instead of $q_h$.

\begin{theorem}[Zero and local explained variation]
\label{thm:weakB}
Assume a joint central limit theorem for the independent stratum means and consistent estimation of their variances, uniformly along the local sequences below. If $m_h=c$ for every $h$, then
\begin{equation}
 J\widehat B_b\ \rightsquigarrow\
 \sum_{k=1}^{H-1}\lambda_k(\chi_{1,k}^2-1),
 \label{eq:zeroB}
\end{equation}
where $\lambda_k$ are the positive eigenvalues of $\bm{\Omega}^{1/2}\bm{A}_q\bm{\Omega}^{1/2}$ and the chi-squared variables are independent. Along $\bm{m}=c\one+\bm{d}/\sqrt J$,
\begin{equation}
 J B_b=\bm{d}^\trans \bm{A}_q\bm{d},\qquad
 J\widehat B_b\ \rightsquigarrow\
 (\bm{Z}+\bm{d})^\trans \bm{A}_q(\bm{Z}+\bm{d})-\tr(\bm{A}_q\bm{\Omega}),\quad \bm{Z}\sim N(0,\bm{\Omega}).
 \label{eq:localB}
\end{equation}
For a full-row-rank $(H-1)\times H$ contrast matrix $\bm{C}_H$ satisfying $\bm{C}_H\one=0$, the statistic
\begin{equation}
 T_J=(\bm{C}_H\bm{\bar Z}_1)^\trans(\bm{C}_H\bm{\widehat D} \bm{C}_H^\trans)^{-1}
                         (\bm{C}_H\bm{\bar Z}_1)
 \label{eq:contexttest}
\end{equation}
converges under the constant-mean null to $\chi^2_{H-1}$. Its local limit is noncentral chi-squared with noncentrality
$(\bm{C}_H\bm{d})^\trans(\bm{C}_H\bm{\Omega} \bm{C}_H^\trans)^{-1}(\bm{C}_H\bm{d})$.
\end{theorem}

We obtain a $J^{-1}$ scale for explained variance near zero, corresponding to $J^{-1/2}$ differences in conditional means. A first-order normal interval for $\widehat B_b$ cannot represent this limit; we can use the centered quadratic limit for diagnostics and local power calculations. To construct an interval covering zero and positive explained variation, we can form a confidence ellipsoid for $\bm{C}_H\bm{m}$ and project it onto $\bm{m}^\trans \bm{A}_q\bm{m}$, a functional invariant to the intercept. The endpoints can be obtained by minimizing and maximizing the quadratic form over the confidence ellipsoid constructed from the estimated contrast covariance. For inference on $B_b$ alone, we can use the mean contrasts directly; joint inference on $(B_b,R_b,V)$ also uses the second moments and their covariance.

For inference on the explained fraction, researchers can specify a minimum total variation and project the joint moment confidence region onto the ratio. Given a confidence region $\mathcal C$ for the first and second moments, they can report
\begin{equation}
 \mathcal I_\eta(v_*)=
 \{B/V:(B,R,V)\in\mathcal C,\ B,R\ge0,\ V=B+R\ge v_*\},
 \label{eq:ratioinversion}
\end{equation}
where $v_*>0$ is a substantively stated minimum total variation. When the true $V\ge v_*$, coverage follows by projection. If the scientific question permits arbitrarily small $V$, the ratio can remain uninformative even when $B$ is accurately estimated. In that case, we report the components and mean contrasts and retain the signed point estimates. When strata contain few sites, we interpret the coverage and contrast-test results using their asymptotic approximations.

\begin{proof}[Proof of Theorem~\ref{thm:weakB}]
The stratum central limit theorem gives $\sqrt J(\bm{\bar Z}_1-\bm{m})\rightsquigarrow \bm{Z}\sim N(0,\bm{\Omega})$ and $J\bm{\widehat D}\to_p\bm{\Omega}$. Since $\bm{A}_q\one=0$, along $\bm{m}=c\one+\bm{d}/\sqrt J$,
\[
 J\bm{\bar Z}_1^\trans \bm{A}_q\bm{\bar Z}_1
 =\{\sqrt J(\bm{\bar Z}_1-\bm{m})+\bm{d}\}^\trans \bm{A}_q
   \{\sqrt J(\bm{\bar Z}_1-\bm{m})+\bm{d}\}.
\]
Subtract the convergent trace correction and apply the continuous mapping theorem. At $\bm{d}=0$, diagonalize $\bm{\Omega}^{1/2}\bm{A}_q\bm{\Omega}^{1/2}$. Its rank is $H-1$ and its positive eigenvalues give \eqref{eq:zeroB}. For $\bm{d}\ne0$, the same calculation gives \eqref{eq:localB}. For the contrast statistic, $\sqrt J \bm{C}_H\bm{\bar Z}_1$ converges to $N(\bm{C}_H\bm{d},\bm{C}_H\bm{\Omega} \bm{C}_H^\trans)$. The covariance is positive definite because $\bm{C}_H$ has full row rank and $\bm{\Omega}\succ0$. Slutsky's theorem then gives the stated central or noncentral chi-squared law.
\end{proof}

\paragraph{Conditions for a uniform local approximation.}
A convenient sufficient condition is independent sampling within each of a fixed number of strata, a uniform $(2+\eta)$ moment bound for the centered first-moment signals for some $\eta>0$, variances uniformly bounded away from zero, and $J_h/J$ bounded away from zero. The stratum Lyapunov bounds then vanish uniformly. Consistency of $s_{1h}^2$ follows from the uniform moment bound. We may condition on independent-pilot augmentation when the signal moment and variance conditions hold with probability tending to one. The approximation requires a fixed number of strata; a growing-$H$ setting requires a different limit argument.

\paragraph{Projection and ratio coverage.}
If an ellipsoid covers $\bm{C}_H\bm{m}$, projecting the ellipsoid onto the quadratic functional includes the value at the covered vector. Because $\bm{A}_q\one=0$, a contrast vector determines $B_b$ uniquely. We therefore obtain coverage by optimizing over the contrast ellipsoid, without requiring a first-order derivative of $B_b$. The same event argument proves coverage of \eqref{eq:ratioinversion} at every true parameter with $V\ge v_*$. We retain and report an empty projected set when projection yields no feasible values. The ratio is undefined at $V=0$.

\subsection{An estimated common participant reference}
\begin{proof}[Proof of Proposition~\ref{prop:referenceBR}]
Let $P_W$ and $P_0$ denote the projections in the preceding proof, and let $\delta$ be the reference perturbation viewed as a function of the site. The population functional expansions are exact:
\begin{align*}
 B(P_M^*)-B(P^*)&=2\langle (P_W-P_0)\tau,\delta\rangle+
                         \|(P_W-P_0)\delta\|_2^2,\\
 R(P_M^*)-R(P^*)&=2\langle (I-P_W)\tau,\delta\rangle+
                         \|(I-P_W)\delta\|_2^2.
\end{align*}
Substitution of $\delta_j=M^{-1}\sum_k\{d_j(\bm{X}_k^*)-\tau_j\}$ gives the reference gradients. Since projections are contractions and
$\E\|\delta\|_2^2=M^{-1}\E_Q\Var_{P^*}\{d_j(\bm{X}^*)\}$, both quadratic terms are $O_p(M^{-1})$. Add the assumed fixed-reference site expansion. The assumed derivative and feature remainder conditions justify replacing empirical fitted derivatives by their population values at the reported scale. Independence of reference and site samples makes the two leading sums independent, hence their covariance matrices add. If the reference comes from the same sampled sites, write its contribution as a sum of site-level reference influences $\bm{\xi}_j$ and use the covariance of $\bm{\phi}_j+\bm{\xi}_j$, including both cross terms.
\end{proof}

\section{Joint allocation proofs}
\begin{proof}[Proof of Theorem~5]
At a fixed site put $d=\tau-f(W)$ and $\varepsilon_i=T_i-\tau$. Then
\[
 L_{f,n}=d^2+\frac{2d}{n}\sum_i\varepsilon_i+
                  \frac1{n(n-1)}\sum_{i\ne k}\varepsilon_i\varepsilon_k.
\]
Conditional centering and independence make the last two terms mean zero and uncorrelated. The linear variance is $4d^2s_j^2/n$. In the ordered-pair sum, only a pair and itself or its reversal contribute to the second moment. Its variance is $2s_j^4/[n(n-1)]$. This uses finite second moments of individual scores; the integrated moments are specified in \textup{(27)}. Total variance over sites with the same $W$ gives \textup{(28)}. The conditional mean $a_f(W)$ does not depend on $n$.

Put $e(W)=a_f(W)-\widetilde a_f(W)$. By \textup{(25)}, the total variance for the pilot-augmented record is
\[
 \Var\{a_f(W)\}-\E e(W)^2+
        \E\{\eta_{f,n(W)}^2(W)/\rho(W)\}.
\]
The first two terms are constant across designs because the predictor and augmentation are fixed across candidate counts. For any $\lambda>0$, pointwise minimization of \textup{(29)} makes its integral no larger than at any competing design. If the pointwise minimizer spends exactly $b_0$, subtracting $\lambda b_0$ and using that a competitor spends at most $b_0$ proves global optimality of its variance. For fixed $n$ with $\eta_{f,n}^2>0$, strict convexity in $r$ yields the clipped square-root solution. If all candidate minimizers are interior, substitution gives $2\sqrt{\lambda c_n\eta_{f,n}^2}$; minimizing this over $n$ proves \textup{(30)}. Zero numerators and binding probability bounds are covered by direct minimization of \textup{(29)}.

If $A_f+e^2>0$ and $c_1>0$, then $c_n\eta_{f,n}^2\sim c_1(A_f+e^2)n\to\infty$, giving a finite integer minimizer in the interior product criterion. Randomization between Lagrangian minimizers preserves the integrated minimizing inequality and can fill an expected-budget gap at an integer switch. A fixed-count or other hard-budget implementation uses its corresponding finite-frame variance.
\end{proof}

\begin{proof}[Proof of Theorem~\ref{thm:minimax}]
For every vector $\bm{v}$, $\bm{v}^\trans \bm{S}(w)\bm{v}\ge0$. Consequently $\bm{S}(w)/r$ is convex in the Loewner order for $r>0$, and integration preserves that order. The largest eigenvalue is convex and monotone in this order, proving convexity of the objective. Equivalently, its epigraph has the matrix constraint $\bm{\Sigma}_D(\rho)\preceq t\bm{I}$.

Introduce a positive-semidefinite matrix multiplier $\bm{H}$ for that constraint and a nonnegative scalar multiplier $\lambda$ for the budget. Minimization of the Lagrangian in $t$ requires $\tr(\bm{H})=1$. The terms depending on $\rho(w)$ are
$\tr\{\bm{H}\bm{D}\bm{S}(w)\bm{D}^\trans\}/\rho(w)+\lambda c(w)\rho(w)$, whose pointwise minimizer is \eqref{eq:spectraldesign}. Matrix complementary slackness gives
$\tr[\bm{H}\{t\bm{I}-\bm{\Sigma}_D(\rho)\}]=0$. Both matrices are positive semidefinite, so $\bm{H}$ has support only on the top-eigenvalue eigenspace at an optimum. The strictly feasible budget and a sufficiently large $t$ give Slater feasibility, establishing the multiplier characterization in finite strata. In the measurable-context formulation, the same characterization follows when an integrable dual optimizer exists. Supplied multipliers satisfying the displayed conditions give a sufficient certificate. Pointwise Lagrangian minimization and complementary slackness directly prove sufficiency without an existence assertion. The bounded probability interval and integrable covariance make the displayed objective finite.
\end{proof}

\paragraph{Dual characterization for finite and measurable contexts.}
We use the theorem's existence result for finite context partitions, including every implemented design in this paper. For arbitrary measurable $W$, the displayed rule certifies global optimality when the required dual optimizers exist; their existence is an additional regularity condition. Discretizing $W$ gives the finite-stratum convex program over the corresponding class of implementable designs. Refining the partition changes the design class and requires reassessing pilot precision.

\paragraph{Regret bound.}
For every feasible $\rho$, the operator-norm difference between estimated and true covariance is at most $\epsilon_0+\sum_hq_h\epsilon/\rho_h\le\epsilon_0+\epsilon/\ell$. Weyl's eigenvalue bound transfers it to the objective. Insert and subtract the estimated objective at the estimated and true optimal designs. The estimated minimizing inequality is at most $\delta_{\rm opt}$, and the two approximation errors yield \eqref{eq:spectralregret}.

\subsection{Allocation with unequal independent cell samples}
\label{sec:cellallocation}
We extend the i.i.d. effect-score construction in Theorem~5 to separate participant counts for independent arm-by-reference cells, using the corrected features in Proposition~\ref{prop:independent_cell_covariance}. We choose site-selection probabilities and participant counts to minimize the variance of the same augmented risk record.

\begin{proposition}[Allocation with cell-specific participant counts]
\label{prop:cell_count_allocation}
Condition on independent training and pilot data and fix $f$. Suppose Proposition~\ref{prop:independent_cell_covariance} applies at each site for every count vector $\mathbf n$ in a prespecified finite menu $\mathcal N(w)$, with $n_h\ge2$ in every cell. Choose the menus, counts, and pairing probabilities before observing evaluation outcomes. Keep the site law, cell response distributions, interventions, and signed reference weights fixed across menus. Assume the following conditional moments and design criteria are finite. Then the corrected feature $L_{f,\mathbf n}$ has conditional mean $a_f(w)=\E\{(\tau-f(w))^2\mid W=w\}$ and variance
\begin{equation}
 \nu_{f,\mathbf n}^2(w)=\Var\{(\tau-f(w))^2\mid W=w\}
 +\E\{4(\tau-f(w))^2v_j(\mathbf n)+d_j(\mathbf n)\mid W=w\},
 \label{eq:cell_count_variance}
\end{equation}
where $v_j(\mathbf n)$ and $d_j(\mathbf n)$ are defined in \eqref{eq:cell_variance_terms}.
Let $\widetilde a_f$ be one fixed pilot augmentation common to all menus, and put
\[
 \eta_{f,\mathbf n}^2(w)=\nu_{f,\mathbf n}^2(w)
       +\{a_f(w)-\widetilde a_f(w)\}^2,
 \qquad
 c_{\mathbf n}(w)=c_0(w)+\sum_{h=1}^K c_h(w)n_h,
\]
with positive cost coefficients. For the record
\[
 \widetilde a_f(W)
 +\frac{C}{\rho(W)}\{L_{f,\mathbf n(W)}-\widetilde a_f(W)\},
\]
assume $C$ is Bernoulli with known probability $\rho(W)$, independently of site responses conditional on $W$. Over $0<\ell(w)\le\rho(w)\le u(w)\le1$ and expected budget $\E\{c_{\mathbf n(W)}(W)\rho(W)\}\le b_0$, choose a pointwise minimizer of
\begin{equation}
 \frac{\eta_{f,\mathbf n}^2(w)}{r}+\lambda c_{\mathbf n}(w)r,
 \quad \mathbf n\in\mathcal N(w),\quad \ell(w)\le r\le u(w),
 \label{eq:cell_count_lagrangian}
\end{equation}
for some $\lambda>0$. If these choices attain the budget with equality, they globally minimize the variance of this record over the stated menus and constraints. For each candidate vector the minimizing probability is
\[
 r_{\mathbf n}(w)=\min\left\{u(w),\max\left\{\ell(w),
 \sqrt{\frac{\eta_{f,\mathbf n}^2(w)}{\lambda c_{\mathbf n}(w)}}\right\}\right\}.
\]
At a zero numerator, the minimizing probability is the lower bound.
\end{proposition}

\begin{proof}
The law of total variance and \eqref{eq:cell_risk_variance} give \eqref{eq:cell_count_variance}. The conditional mean is the same for every menu. Consequently, conditional on the pilot, the augmented record has variance
\[
 \Var\{a_f(W)\}
 -\E\{a_f(W)-\widetilde a_f(W)\}^2
 +\E\!\left\{\frac{\eta_{f,\mathbf n(W)}^2(W)}{\rho(W)}\right\}.
\]
The first two terms are constant across designs. Minimizing \eqref{eq:cell_count_lagrangian} over $r$ at each candidate count gives the clipped rule; a zero numerator gives $r=\ell(w)$ because the cost coefficient and multiplier are positive. Integrate the pointwise minimizing inequality. The proposed design spends $b_0$, whereas any feasible competitor spends at most $b_0$; subtracting $\lambda b_0$ proves the variance inequality.

If tied menu--probability pairs are randomized before evaluation outcomes, independently of site responses conditional on $W$ and the pilot, record the selected menu $N$ and use its conditional pairing probability $\rho(W,N)$ in the inverse-probability record. The expectation in the last variance term then includes $N$. Both the expected cost and Lagrangian value average over the mixture, preserving the same proof whenever the mixture attains the expected budget.
\end{proof}

For a fixed site, increasing $n_h$ makes that cell's $w_h^2\sigma_h^2/n_h$ contribution to $v_j$, its $2w_h^4\sigma_h^4/\{n_h(n_h-1)\}$ contribution to $d_j$, and every cross-cell term containing $n_h$ nonincreasing. We therefore incorporate the reference weights and arm imbalance when calculating the variance for each candidate count vector. Honduras uses two native-composition cells; IST uses four arm-by-age cells with fixed reference weights. The IST information-deletion experiment keeps those realized cell counts fixed and varies site selection. For a future experiment under the same response-sampling conditions, we can use the proposition to plan participant counts.

For estimated inputs, the finite-partition regret argument in Supplement Section~\ref{sec:designproofs} also applies uniformly over finite menus. With known stratum masses, fixed costs and probability bounds, let $\ell=\min_s\ell_s>0$, where $s$ indexes context strata. If the error in $\eta_{f,\mathbf n}^2$ is at most $\epsilon$ for every stratum and menu, a design optimized to tolerance $\delta$ has objective regret at most $2\epsilon/\ell+\delta$. To obtain the bound, we control the criterion error uniformly by $\epsilon/\ell$ and use the minimizing inequality for the estimated criterion. The regret bound controls the increase in the variance criterion caused by errors in the planning inputs. The corrected feature conditions and the known recorded probabilities determine the conditional mean of the evaluation record.

\section{Specified references and finite-site implementation}
\label{sec:targetextensions}
\begin{proposition}[Transfer to a specified site reference]
\label{prop:sitereference}
Let $\Pi_j$ be specified before experimental-pattern assignment and applied to both arms. Suppose the randomization and first- and second-moment feature equations in Assumption~1 hold for $\theta_{aj}^{\Pi}=\int m_{aj}\,d\Pi_j$, conditional on the site reference and independent feature-training data. Suppose the target moments agree across every pattern in which a feature is used. Then the decomposition, projection, identification, and signal-mean identities hold for $\tau_j^{\Pi}=\theta_{1j}^{\Pi}-\theta_{0j}^{\Pi}$. The inference and allocation results transfer under their respective moment, feature-error, and independence conditions.
\end{proposition}
\begin{proof}
Condition on the site reference and independent training data. The randomized-observation calculation in \textup{(16)} gives $\E(Z_1\mid W)=\E(\tau^{\Pi}\mid W)$ and $\E(Z_2\mid W)=\E\{(\tau^{\Pi})^2\mid W\}$. The projection identities and subsequent proofs use these equations and their stated additional conditions; participant references can differ across sites.
\end{proof}

\begin{proposition}[Finite-site projection correction]
\label{prop:finiteprojection}
Conditional on the site response populations, the recorded basis and any independent feature-training data, and evaluation-arm counts, suppose $\E(\bm{z})=\bm{\tau}$, $\Cov(\bm{z})=\bm{C}$, and $\E(\bm{\widehat C})=\bm{C}$. Let $\bm{P}$ be a fixed orthogonal projection containing the intercept and $\bm{P}_0=\one\one^\trans/J$. Then
\begin{align*}
 \widehat B_F&=\{\bm{z}^\trans(\bm{P}-\bm{P}_0)\bm{z}-\tr[(\bm{P}-\bm{P}_0)\bm{\widehat C}]\}/J,\\
 \widehat R_F&=\{\bm{z}^\trans(\bm{I}-\bm{P})\bm{z}-\tr[(\bm{I}-\bm{P})\bm{\widehat C}]\}/J
\end{align*}
are unbiased for $B_F=\bm{\tau}^\trans(\bm{P}-\bm{P}_0)\bm{\tau}/J$ and $R_F=\bm{\tau}^\trans(\bm{I}-\bm{P})\bm{\tau}/J$. Their sum is the corresponding corrected finite-site total variance.
\end{proposition}
\begin{proof}
For a fixed symmetric $\bm{A}$, $\E(\bm{z}^\trans \bm{A} \bm{z})=\bm{\tau}^\trans \bm{A}\bm{\tau}+\tr(\bm{A}\bm{C})$. Apply this identity to $\bm{P}-\bm{P}_0$ and $\bm{I}-\bm{P}$; addition gives $\bm{I}-\bm{P}_0$.
\end{proof}
For independent arm samples from each site response population, $\bm{C}$ is diagonal and $\widehat C_{jj}=s_{1j}^2/n_{1j}+s_{0j}^2/n_{0j}$ is unbiased. When sites share estimated coefficients or a participant reference, we use the full covariance matrix to account for the resulting dependence. We establish unbiasedness from the expectation identity using only the specified means and covariances, and assess confidence-region coverage through a separate distributional argument.

Let $\bm{H}$ have orthonormal columns spanning the basis, with first column $\one/\sqrt J$, and define $\bm{\gamma}=\bm{H}^\trans\bm{\tau}/\sqrt J$, $q_2=\|\bm{\tau}\|^2/J$. Then $B_F=\sum_{k\ge2}\gamma_k^2$, $R_F=q_2-\|\bm{\gamma}\|^2$, and $V_F=q_2-\gamma_1^2$. Pythagoras gives $q_2\ge\|\bm{\gamma}\|^2$. For unbounded log outcomes, we constrain the primitive moments using this inequality; the upper bound used for binary outcomes pertains to their bounded support. By projecting the primitive moment confidence region, we obtain intervals with its coverage. In the application, we construct that region using a working normal approximation and an estimated covariance matrix.

For a fixed finite cohort under complete randomization, the Neyman subtraction instead obeys
\[
 \E_A\!\left(\widehat\tau^2-s_1^2/n_1-s_0^2/n_0\right)
   =\tau_N^2-S_{\tau,N}^2/N.
\]
For the fixed-cohort target, the Neyman subtraction differs from the squared mean effect by the displayed finite-cohort effect-variance term. To apply Proposition~1, we also require equal conditional means for the independent unit scores. We use independent feature training to control dependence from estimation and specify the sampling law separately to distinguish response-population sampling from treatment randomization in a fixed finite cohort.

\section{Further reference and multiple-target results}
\label{sec:furtherextensions}
We account for uncertainty in an estimated common participant reference and derive an allocation rule for several simultaneous risk targets.

\subsection{Uncertainty in a common participant reference}
\label{sec:commonreference}
When we estimate the participant reference from a common sample, the standardized site effects share the resulting estimation error. Let $\bm{X}_1^*,\ldots,\bm{X}_M^*$ be an independent reference sample and write $d_j(\bm{x})=m_{1j}(\bm{x})-m_{0j}(\bm{x})$, so $\tau_j=\int d_j\,dP^*$. Substituting its empirical distribution changes every site effect by
$\delta_j=M^{-1}\sum_{k=1}^M\{d_j(\bm{X}_k^*)-\tau_j\}$. We derive the first derivatives of the two variance components with respect to that shared distribution as
\begin{align}
 \psi_B(\bm{x})&=2\E_Q[(m(W)-\mu)\{d_j(\bm{x})-\tau_j\}],\nonumber\\
 \psi_R(\bm{x})&=2\E_Q[(\tau_j-m(W))\{d_j(\bm{x})-\tau_j\}],\qquad
 \psi_V=\psi_B+\psi_R.
 \label{eq:referencegradients}
\end{align}
All three functions have mean zero under $P^*$. In the residual-component gradient, the centered response contrast is weighted by the site effect remaining after conditioning on site characteristics. Estimating the common reference therefore contributes to residual-component uncertainty even when the conditional mean is estimated accurately.

\begin{proposition}[Two sources of uncertainty]
\label{prop:referenceBR}
Suppose the estimator at a fixed reference has a site influence expansion with vector $\bm{\phi}_j=(\phi_{Bj},\phi_{Rj})^\trans$. Suppose empirical-reference substitution has the first-order derivative \eqref{eq:referencegradients}, finite integrated squared response contrasts, and remainder negligible relative to the standard error of every reported contrast. Then
\begin{equation}
 \begin{pmatrix}\widehat B-B\\\widehat R-R\end{pmatrix}
 =\frac1J\sum_j\bm{\phi}_j+
   \frac1M\sum_k\begin{pmatrix}\psi_B(\bm{X}_k^*)\\\psi_R(\bm{X}_k^*)\end{pmatrix}
   +\text{remainder}.
 \label{eq:referenceexpansion}
\end{equation}
The first-order covariance is the sum of the site covariance divided by $J$ and the reference covariance divided by $M$. The empirical-reference functional itself has a quadratic remainder $O_p(M^{-1})$. Generated response fits must additionally satisfy the feature and derivative consistency conditions used for \eqref{eq:referenceexpansion}.
\end{proposition}

For the explained component the site contribution is ordinarily of order $J^{-1}$; for the residual component it can be of order $(Jr_J)^{-1}$. The reference contribution is of order $M^{-1}$ for both. When reference uncertainty dominates, we improve precision by increasing the reference sample; increasing the number of paired sites alone then has little effect. Near a zero component its derivative can vanish; a local quadratic analysis with an estimated reference also requires the reference perturbation. We derive the nondegenerate first-order expansion in this proposition and the boundary limit for a fixed reference in Theorem~\ref{thm:weakB}. When we estimate the reference from the same sampled sites, we also include the covariance between the site and reference contributions, as detailed in the proof.

\subsection{Allocation for multiple contrasts}
\label{sec:multitarget}
An experiment may need to estimate a vector of regional residual risks or several outcomes. Minimizing the variance of their average can leave one scientifically important contrast imprecise. Let $\bm{L}\in\mathbb R^K$ be an unbiased vector of paired-site features with conditional mean $\bm{a}(W)$, conditional covariance $\bm{S}(W)$, and target $\bm{\zeta}=\E \bm{L}$. Let $\bm{D}$ be a prespecified nonsingular diagonal matrix converting the coordinates to scientifically comparable units. Under oracle augmentation the scaled influence covariance is
\begin{equation}
 \bm{\Sigma}_D(\rho)=\bm{D}\Var\{\bm{a}(W)\}\bm{D}^\trans+
               \E\{\bm{D}\bm{S}(W)\bm{D}^\trans/\rho(W)\}.
 \label{eq:multicov}
\end{equation}
The largest eigenvalue is the largest variance among unit-length linear contrasts of $\bm{D}\bm{\zeta}$. We minimize this largest variance to improve precision for the least precise standardized contrast.

\begin{theorem}[Allocation for the least precise contrast]
\label{thm:minimax}
Suppose $W$ takes values in a fixed finite partition, $\E\|\bm{S}(W)\|<\infty$, costs are positive, and the probability bounds have a positive lower bound. Under a strictly feasible budget, minimizing $\lambda_{\max}\{\bm{\Sigma}_D(\rho)\}$ is a convex problem. At a binding optimum with a positive budget multiplier $\lambda$, there exists $\bm{H}\succeq0$, $\tr(\bm{H})=1$, supported on the eigenspace of the largest eigenvalue of $\bm{\Sigma}_D(\rho^*)$, such that
\begin{equation}
 \rho^*(w)=\min\left[u(w),\max\left\{\ell(w),
   \sqrt{\frac{\tr\{\bm{H}\bm{D}\bm{S}(w)\bm{D}^\trans\}}{\lambda c(w)}}\right\}\right].
 \label{eq:spectraldesign}
\end{equation}
These conditions, the budget constraint, and the complementary-slackness conditions are also sufficient. If the numerator is zero and $\lambda>0$, the lower bound is optimal. The finite-stratum problem is a semidefinite-constrained convex program.
\end{theorem}

The matrix $\bm{H}$ places weight on contrasts that are least precise under the chosen design. A simple largest eigenvalue gives a rank-one $\bm{H}$; a tie can require a mixture of directions. The theorem extends the fixed importance-matrix calculation after Theorem~4 by determining those weights from the precision objective. The conditional covariance of the paired features incorporates lower-level measurement variation and shared-outcome correlations.

For finite strata, let $\bm{S}_h$ be the true conditional covariance and $\bm{\widehat S}_h$ a positive-semidefinite pilot estimate. If $\max_h\|\bm{D}(\bm{\widehat S}_h-\bm{S}_h)\bm{D}^\trans\|_{\mathrm{op}}\le\epsilon$ and the baseline scaled covariance is estimated within $\epsilon_0$ in operator norm, then the plug-in objective differs uniformly from the true objective by at most $\epsilon_0+\epsilon/\ell$. A design optimized to tolerance $\delta_{\rm opt}$ has objective regret at most
\begin{equation}
 2(\epsilon_0+\epsilon/\ell)+\delta_{\rm opt}.
 \label{eq:spectralregret}
\end{equation}
We obtain this bound by applying the eigenvalue perturbation inequality and the minimizing inequality for the estimated criterion. A small probability floor increases the effect of pilot covariance error on the objective, particularly for several residual targets. Alternatively, we can minimize the largest objective over a pilot confidence set of covariance matrices. On the coverage event for that set, the resulting objective bounds the achieved criterion simultaneously over those matrices.

\section{Simulation design and numerical results}
\label{sec:numericalprotocol}
We use NumPy's \texttt{SeedSequence(20260923)} for the numerical experiments and IST analysis, with four independent child streams, in order: IST, the main simulation, the allocation simulation, and the projection simulation. We run each stochastic experiment for 1,000 replications per setting. We save replicate records, summary CSV files, and seeds in the replication package. We calculate the IST features from the supplied raw file and generate the simulation data from the distributions below.

\subsection{Bounded data-generating mechanisms}
We generate sites in equally weighted strata $h=1,2,3$ with $w_h=(-1,0,1)$. At each site, we independently draw $U,T$ taking values $-1,+1$ with equal probabilities and set
\[
 \theta_0=.40+\gamma w_h+.04U,\qquad
 \tau=m_h+\sqrt{v_h}T,\qquad \theta_1=\theta_0+\tau.
\]
The parameter grid is
\begin{center}\small
\begin{tabular}{lccc}
\toprule
Mechanism & $(m_1,m_2,m_3)$ & $(v_1,v_2,v_3)$ & $\gamma$\\
\midrule
Mean modifier & $(-.04,.04,.12)$ & $(.004,.004,.004)$ & .04\\
Variance modifier & $(.04,.04,.04)$ & $(.0016,.0064,.0168)$ & .04\\
Prognostic modifier & $(.04,.04,.04)$ & $(.008267,.008267,.008267)$ & .16\\
Zero heterogeneity & $(.04,.04,.04)$ & $(0,0,0)$ & .04\\
\bottomrule
\end{tabular}\end{center}
For the prognostic mechanism, we use the double-precision representation of $124/15000$. All arm probabilities lie strictly inside $[0,1]$. In each arm, we independently sample $n_h=(80,40,20)$ Bernoulli outcomes with the specified site probability. We compute the first- and second-moment features from their sums. In these two-point mechanisms, the conditional squared residual is constant within each latent stratum, while its estimated feature varies because of participant sampling. We evaluate the allocation criterion with this sampling variation included, rather than an oracle criterion based only on latent squared residuals.

We use an independent pilot with 60 paired sites per stratum to estimate sample-mean predictors of $S_0,S_1,Q$. We generate $J_h=J/3$ evaluation sites per stratum and select paired sites independently with Bernoulli probability $\rho$. For each remaining site, we select one arm with equal probabilities. Marginal arm observation probability is $(1+\rho)/2$. We hold the number of participants in each observed arm fixed across observation patterns; whole-site assignment provides one arm sample of that size. We thereby compare information acquisition at fixed arm sample sizes; to evaluate a different participant budget, we would update the feature covariance and cost inputs together.

The four evaluation settings are $(J,\rho)=(300,.2),(300,.5),(900,.2),(300,0)$. Under $\rho=0$, we estimate and evaluate only $B_b$; we do not treat the pilot's paired second moments as information identifying a residual effect variance in the pure-cluster evaluation population. Complete-pair estimation uses $Z_1=C\widetilde\tau/\rho$, $Z_2=CQ/\rho$ with all $J$ records and the finite-stratum correction. The uncorrected method replaces $Q$ by $\widetilde\tau^2$. Each estimator's theoretical signal and uncertainty are recomputed under its own observation pattern.

The leading influence vectors for conditional fixed-stratum inference are centered within each stratum. With $m_h$ and $q_{2h}$ replaced by their signal means, they are
\[
 \phi_{Bj}=2(m_h-\mu)(Z_{1j}-m_h),\qquad
 \phi_{Rj}=Z_{2j}-q_{2h}-2m_h(Z_{1j}-m_h).
\]
The estimated variance is $\sum_hq_h^2s_{\phi,h}^2/J_h$. This Wald approximation is diagnostic at a zero explained component, where its first-order derivative vanishes. The primary boundary intervals use the primitive region described in Section~\ref{sec:projectioninference}.

\subsection{Coverage of projection confidence intervals}
Using an independent child stream, we repeat each mechanism at $J=300,900$ and $\rho=.2$ with a fresh pilot and evaluation sample in each replication. We evaluate coverage of the projected intervals constructed from the normal primitive moment confidence region used in the application, including the empty-region fallback. Table~\ref{tab:projectionperformance} reports coverage of all three components together, coverage of $R$, its mean interval width, and the empty-region frequency. Monte Carlo standard error of a proportion $p$ is $\sqrt{p(1-p)/1000}$.

\begin{table}[htbp]\centering\small
\caption{Implemented projection inference, 1,000 replications per setting and paired probability .2. Joint refers to simultaneous inclusion of $B_b,R_b,V$. Width is in squared-outcome units.}
\label{tab:projectionperformance}
\begin{tabular}{lrrrrr}
\toprule
Scenario & $J$ & Joint cov. & Cov. $R$ & Width $R$ & Empty rate \\
\midrule
Mean modifier & 300 & 1.000 & 1.000 & 0.0268 & 0.001 \\
Mean modifier & 900 & 0.999 & 1.000 & 0.0179 & 0.000 \\
Prognostic modifier & 300 & 1.000 & 1.000 & 0.0277 & 0.001 \\
Prognostic modifier & 900 & 0.999 & 1.000 & 0.0198 & 0.000 \\
Variance modifier & 300 & 0.998 & 1.000 & 0.0297 & 0.002 \\
Variance modifier & 900 & 0.997 & 1.000 & 0.0201 & 0.000 \\
Zero heterogeneity & 300 & 0.992 & 1.000 & 0.0735 & 0.060 \\
Zero heterogeneity & 900 & 0.991 & 1.000 & 0.0296 & 0.021 \\
\bottomrule
\end{tabular}

\end{table}

We evaluate coverage under the stated bounded mechanisms and site counts. For the application, we interpret the regional intervals using the primitive-moment normal approximation at the smaller sample size of its smallest geographic stratum, which contains fewer sites than a simulation stratum.

\subsection{Corrected measurement and information rates}
\begin{table}[htbp]\centering\footnotesize
\caption{Residual-variation comparison at $J=300$, $\rho=.2$ (1,000 replications). Bias and RMSE refer to the same residual target. The primary interval procedure is evaluated separately in Table~\ref{tab:projectionperformance}.}
\label{tab:allcomparators}
\begin{tabular}{llrr}
\toprule
Scenario & Method & Bias & RMSE \\
\midrule
Mean modifier & Fusion & -0.00004 & 0.00445 \\
Mean modifier & Complete pair & 0.00005 & 0.00367 \\
Mean modifier & Uncorrected & 0.01397 & 0.01449 \\
Variance modifier & Fusion & 0.00028 & 0.00464 \\
Variance modifier & Complete pair & -0.00007 & 0.00465 \\
Variance modifier & Uncorrected & 0.01413 & 0.01515 \\
Prognostic modifier & Fusion & -0.00019 & 0.00416 \\
Prognostic modifier & Complete pair & 0.00011 & 0.00423 \\
Prognostic modifier & Uncorrected & 0.01372 & 0.01448 \\
Zero heterogeneity & Fusion & 0.00007 & 0.00320 \\
Zero heterogeneity & Complete pair & 0.00001 & 0.00309 \\
Zero heterogeneity & Uncorrected & 0.01417 & 0.01460 \\
\bottomrule
\end{tabular}

\end{table}
\begin{table}[htbp]\centering\footnotesize
\caption{Corrected fusion: precision as the number of sites and the pairing probability vary. Rows are extracted from the completed simulation summaries for the four reported mechanisms.}
\label{tab:allrates}
\begin{tabular}{lrrrr}
\toprule
Scenario & $J$ & $\rho$ & RMSE $B$ & RMSE $R$ \\
\midrule
Mean modifier & 300 & 0.2 & 0.00139 & 0.00445 \\
Mean modifier & 300 & 0.5 & 0.00129 & 0.00257 \\
Mean modifier & 900 & 0.2 & 0.00085 & 0.00253 \\
Variance modifier & 300 & 0.2 & 0.00031 & 0.00464 \\
Variance modifier & 300 & 0.5 & 0.00023 & 0.00281 \\
Variance modifier & 900 & 0.2 & 0.00010 & 0.00264 \\
Prognostic modifier & 300 & 0.2 & 0.00026 & 0.00416 \\
Prognostic modifier & 300 & 0.5 & 0.00021 & 0.00261 \\
Prognostic modifier & 900 & 0.2 & 0.00009 & 0.00234 \\
Zero heterogeneity & 300 & 0.2 & 0.00019 & 0.00320 \\
Zero heterogeneity & 300 & 0.5 & 0.00015 & 0.00199 \\
Zero heterogeneity & 900 & 0.2 & 0.00006 & 0.00181 \\
\bottomrule
\end{tabular}

\end{table}
We provide all evaluated mechanisms, methods, and targets, including $V$, in the machine-readable summaries. We report four mechanisms to compare the effects of information acquisition and measurement correction, including settings in which complete-pair estimation has smaller RMSE. We identify the residual target from the cross-arm moment and compare finite-sample RMSE using each method's observation probabilities and feature variances.

\subsection{Independent-pilot allocation experiment}
For the allocation experiment, we use the residual-variance mechanism, an independent pilot of 180 paired sites, and 600 evaluation sites, equally divided among the contexts. From the pilot, we estimate $f_h$ by the mean effect and estimate $a_h,\nu_h^2$ from $L_f=Q-2f_h\widetilde\tau+f_h^2$. We minimize the estimated conditional-variance criterion at expected pairing fraction .30, equal incremental costs, and bounds [.10,.90]. We evaluate estimation of the absolute risk
$\frac13\sum_h\{v_h+(m_h-f_h)^2\}$ conditional on that replication's pilot. We use the same evaluation features and uniform random draws for paired-site selection under both designs, so the comparison is paired across replications. In each replication, we re-estimate the pilot inputs, including the conditional variances, using the independent pilot outcomes. Across replications, the estimated rule has mean pairing probabilities of about $(.112,.252,.535)$.

\subsection{Joint site and participant allocation example}
\label{sec:extensionnumerics}
For two contexts with masses $(0.4,0.6)$, we specify
$\bm{A}=(0.01,0.06)$, $\bm{D}=(0.4,0.8)$, $\bm{E}=(1,2)$,
fixed costs $(5,8)$, participant costs $(0.2,0.3)$,
and candidate participant counts $2,\ldots,40$.
The expected budget is 2 per eligible site and the pairing bounds are
$[0.02,0.90]$. All candidate Lagrangian solutions are interior.
We obtain counts $(35,22)$ and pairing probabilities
$(0.09222,0.17778)$ from the joint rule, with design variance 0.43643.
The best common count is 24, with variance 0.44152.
Compared with requiring a common count, allowing context-specific
participant counts reduces the variance by 1.15\% for these contexts,
costs, and budget. We verify the optimum by
exhaustively enumerating all $39^2$ count pairs and independently
optimizing the pairing probabilities for each pair.
We jointly choose which sites to pair and how many participants to
measure, using the calculation implemented in
\path{replication/code/v28/theory_checks.py}. We specify the costs and budget in normalized planning units. These comparisons use the
design-dependent criterion $\E\{\nu_{f,n(W)}^2(W)/\rho(W)\}$ under
oracle augmentation; the common $\Var\{a_f(W)\}$ term does not affect
the ranking.

To check a binding probability constraint, we retain the same contexts,
variance inputs, count menus, costs, budget, and upper bound, and raise
only the lower bound to 0.10. Direct minimization of
\textup{(29)} gives counts $(32,22)$ and probabilities
$(0.10000,0.17626)$, with criterion 0.43680 and expected cost 2.
The lower bound binds in the first context. The interior product rule
would still choose $(35,22)$; thus clipping can change the optimal count.
For each of the $39^2$ count pairs, we independently minimize the
criterion over its feasible one-dimensional probability interval,
including both endpoints. The exhaustive minimum agrees with the
Lagrangian result to $6\times10^{-17}$, and the primal--dual gap is
below $2\times10^{-16}$. The best common count remains 24, with
criterion 0.44152. We retain the original interior comparison and save
both checks, including all enumerated candidates, in
\path{replication/v38/results/joint_allocation_boundary_checks.json}.
The reproducible calculation is
\path{replication/v38/code/check_joint_allocation.py}; it checks budget
equality explicitly and reports an integer-menu cost jump when tied
choices would require mixing. In both examples, we constrain the expected cost of Bernoulli
site selection.

\section{Application details}
\paragraph{Targets and measurement records.}
We use the same corrected cross-arm features for the population
results and the application calculations, holding different quantities
fixed in each analysis. Table~\ref{tab:targetmap} records the targets,
the quantities held fixed, and the corresponding inferential results. The predictor $m$ is the population
conditional mean; $f_b$ is its projection onto a specified span; a pilot
predictor $f$ is fixed before evaluation. We can estimate the risk of a predictor fixed before evaluation
without requiring that predictor to be consistent for $m$.

Theorem~1 describes observation requirements in
the projected-law experiment. Propositions~1 and
\ref{prop:independent_cell_covariance} recover the needed first and
quadratic moments from finite independent response samples. Using the known pattern probabilities, we construct the unbiased
augmented signals in \textup{(16)}, then average and transform
those moments to estimate the chosen target.

\begin{table}[htbp]\centering\small
\caption{Targets, evaluation records, and inferential results. All
corrected products use distinct independent response units under the
specified response-sampling model.}
\label{tab:targetmap}
\begin{tabularx}{\textwidth}{R{.25\textwidth}Y Y}
\toprule
Target & Evaluation and quantities held fixed & Result\\
\midrule
Population $R$ & Sampled sites; conditional-mean predictor estimated
on separate site folds & Theorem~3: total-site and
paired-site information rates\\
Population $B_b,R_b$ & Prespecified explanatory span; fixed stratum
mixture when using a partition & Theorem~3 and
Proposition~2\\
IST $\mathcal R(f)$ & Independent pilot predictor and augmentation;
151-hospital evaluation frame & Conditional deletion reconstruction;
Theorem~4 supplies the planning criterion\\
Honduras $B_F,R_F$ & 40-village frame and recorded attribute basis;
household response sampling & Proposition~\ref{prop:finiteprojection}:
unbiased trace correction; working primitive intervals\\
\bottomrule
\end{tabularx}
\end{table}

For a future hybrid experiment, investigators also fix the site law,
participant reference, and intervention versions across assignment
patterns. A candidate participant-count menu changes measurement
precision while preserving these targets, as specified in
Proposition~\ref{prop:cell_count_allocation}.

\label{sec:istprotocol}
\subsection{IST data definitions, corrected features, and allocation}
\paragraph{Source and analysis frame.}
The raw file is \texttt{IST\_corrected.csv}, included in the reproducibility package and attributed to Edinburgh DataShare \citep{ISTdata}. Its SHA-256 is recorded in \texttt{ist\_results.json} and the package manifest. We read it with Latin-1 encoding. The source contains 19,435 records and 466 distinct values of \texttt{HOSPNUM}. A patient is assigned aspirin when \texttt{RXASP=Y}; the comparison is \texttt{RXASP=N}. Age group is \texttt{AGE>80}.

For 14-day mortality, we use \texttt{ID14} when \texttt{SET14D=1} and otherwise mark the outcome missing. For the six-month composite, \texttt{OCCODE} in $\{1,2\}$ gives 1 and in $\{3,4\}$ gives 0; other codes are missing. The trial's factorial design assigned heparin avoidance, 5,000 IU twice daily, and 12,500 IU twice daily with planned marginal probabilities $(.50,.25,.25)$, crossed with aspirin assignment \citep{IST1997}. The public field \texttt{RXHEP} codes these as N, L, and M, with H denoting the 12,500-IU arm in the trial's initial pilot phase \citep{Sandercock2011}. We average over the randomized heparin law, allowing an aspirin--heparin interaction. For age group $g$ at hospital $j$, write $\eta_{agj}(h)$ for the response mean under aspirin arm $a$ and heparin assignment $h$, and let $\lambda=(.50,.25,.25)$. Our standardized arm target is
\[
 \theta_{aj}^{*}=\sum_g\pi_g^*\sum_{h\in\{N,L,M\}}\lambda_h\eta_{agj}(h),
 \qquad \tau_j^{\mathrm{benefit}}=\theta_{0j}^{*}-\theta_{1j}^{*}.
\]
Under the response-population sampling and availability conditions below, each aspirin-arm mean averages over this same factorial law. We do not condition the target on either arm's realized heparin proportions or post-stratify to their finite-cohort average. Hospital-specific patient distributions within age groups remain part of the target. The information-deletion calculation instead holds the observed patient arrays fixed and reconstructs their corrected risk estimate.

We require at least two observed outcomes in each hospital by aspirin by age cell for both endpoints, applying this rule once to form the common 226-hospital frame. Inclusion depends on counts and outcome availability, not on response values or estimated treatment effects. We assume outcome-independent availability within those cells for complete-case means and weight each retained hospital equally. For regional component inference, we treat hospitals as independent response-population draws within the three strata, condition on their counts 70, 129, and 27, and fix mixture masses at $J_h/226$. These intervals concern response-population sampling of hospitals; exact treatment-randomization inference for a fixed vector of 226 hospital effects has a different sampling target. In the deletion experiment, we instead condition on the realized 151-hospital evaluation frame and the pilot.

We start from the public frame of 466 hospital identifiers and 19,435 patients and exclude 240 identifiers and 3,388 patients through the eligibility rule. We retain a 226-hospital frame with 16,047 patients, of whom 5,541 belong to the 75 pilot hospitals and 10,506 to the 151 evaluation hospitals. We interpret the estimates within this eligible analysis population, defined by response availability.

Other Europe is coded by the country labels \texttt{ITAL, SWIT, POLA, NETH, SWED, NORW, SPAI, CZEC, PORT, BELG, AUST, GREE, HUNG, SLOK, FINL, EIRE, SLOV, DENM, ROMA, FRAN}. United Kingdom is \texttt{UK}; all remaining labels form the third group. We use these descriptive groups without selecting site moderators through treatment-interaction tests.

\begin{table}[htbp]\centering\small
\caption{Regional residual variation in age-standardized 14-day aspirin benefit on the common 226-hospital frame. Signed variance estimates are retained. Upper limits use the simultaneous primitive construction across regions for this endpoint; they rely on the primitive approximation.}
\label{tab:istcontext}
\begin{tabular}{lrrr}
\toprule
Region & Sites & Residual variance & Upper limit \\
\midrule
United Kingdom & 70 & -0.00095 & 0.00338 \\
Other Europe & 129 & 0.00555 & 0.01381 \\
Other countries & 27 & -0.00222 & 0.00531 \\
\bottomrule
\end{tabular}

\end{table}

\paragraph{Exact age-mixture products.}
For group $g\in\{0,1\}$ and arm $a$, let $\bar Y_{agj}$ be the complete-case mean and
\[
 U_{agj}=\frac{(\sum_iY_{iagj})^2-\sum_iY_{iagj}^2}
                  {n_{agj}(n_{agj}-1)}.
\]
With fixed participant-reference weights $(\pi_0^*,\pi_1^*)=(.75,.25)$, distinct from the whole-site arm-assignment probabilities $\pi_a(W)$,
\[
 S_{aj}=\sum_g\pi_g^*\bar Y_{agj},\qquad
 U_{aa,j}=\sum_g\pi_g^{*2}U_{agj}
         +2\pi_0^*\pi_1^*\bar Y_{a0j}\bar Y_{a1j},
\]
and $U_{01,j}=S_{0j}S_{1j}$. Under independent patient sampling within the response-population model and the stated randomized arm allocation, these are unbiased for the corresponding age-standardized moments conditional on cell sample sizes and the hospital. We compute all four cells from disjoint patients and use the same specified age reference in both arms. When we change the participant reference, we compare different site-effect targets and interpret their variance difference as a target comparison rather than an additive composition contribution.

\paragraph{Pilot and information deletion.}
We draw a stratified independent pilot comprising 23 UK hospitals, 43 other-European hospitals, and nine hospitals from the third region. We save the selected IDs, remaining evaluation IDs, and seed stream in the result file. The enrollment-size strata use cutoffs 38 and 77 computed from baseline counts in the entire eligible frame; evaluating these covariates does not use response values. The pilot's regional mean predictor is fixed during deletion. Its corrected squared-error features are grouped by enrollment-size stratum to estimate an augmentation mean and variance. For planning, we impose a variance floor of $10^{-6}$ and restrict pairing probabilities to [.10,.90]. Costs are one per additional paired hospital. Both designs have expected evaluation cost $151\times.30=45.3$ paired-hospital units plus the same fully paired pilot.

\begin{table}[htbp]\centering\small
\caption{Independent-pilot allocation for the 151 evaluation hospitals. Share refers to the evaluation frame; pilot SD is the sample standard deviation of the corrected prediction-error feature. Expected paired counts sum to 45.3.}
\label{tab:istallocation}
\begin{tabular}{lrrrr}
\toprule
Enrollment stratum & Share & Pilot SD & Pair probability & Expected pairs \\
\midrule
Small & 0.311 & 0.0234 & 0.401 & 18.8 \\
Middle & 0.351 & 0.0236 & 0.403 & 21.4 \\
Large & 0.338 & 0.0054 & 0.100 & 5.1 \\
\bottomrule
\end{tabular}

\end{table}

In each deletion replicate, we retain one arm, selected with equal probabilities, in each nonpaired hospital. The code first creates a dictionary of the available arm arrays, then computes means and products only from those arrays. We estimate all outcome-dependent functions in the independent pilot, so evaluation requires no patient-level regression refitting or access to hidden-arm outcomes. To check this restriction, we arbitrarily change the hidden outcomes and verify that the observable features remain unchanged.

The scientific target is the absolute risk of the pilot predictor in the eligible evaluation population. The realized full-information corrected estimate is the deletion reference, $0.0032237$. Conditional on the fixed frame and pilot, the augmented deletion estimator is exactly unbiased for this reference. We use the pilot to fix the predictor, augmentation, and allocation inputs. The evaluation-frame second moment is estimated from the retained experimental records. The two exact conditional variances are $1.7110\times10^{-5}$ (uniform) and $1.1248\times10^{-5}$ (pilot allocation). Population sampling and uncertainty of the reference are separate from this design-deletion calculation. The simulation evaluates population-risk error with a known target.

The small, middle, and large enrollment strata contain 27, 22, and 26 pilot hospitals and 47, 53, and 51 evaluation hospitals. Their pilot sample variances of $L_f$ are .00054908, .00055565, and .00002949; none reaches the $10^{-6}$ planning floor. The allocation masses are therefore $(47,53,51)/151$, whereas the regional heterogeneity mixture uses $(70,129,27)/226$. The protocol file \texttt{ist\_protocol\_inputs.json} records these counts, the fixed age weights, pairing probabilities, common budget, and unrounded RMSE values. We calculate the percentage reductions in the main text from the unrounded results, rather than the displayed table values. These diagnostics concern the specified pilot; the simulation experiment re-estimates its pilot inputs in each replication.

\subsection{Honduras measurement protocol}
\label{sec:hondurasprotocol}
\paragraph{Frame, source variables, and observed responses.}
We use the supplied six-country file for \citet{Banerjee2015}, \path{pooled_hh.dta}, beginning with its 2,403 Honduras households in 80 villages. We retain the 40 intervention villages, selected by \texttt{country=3} and nonmissing \path{invillage_treat}; the 794 households in pure-control villages are excluded. Within this frame, \texttt{treatment} equals \path{invillage_treat} and records household assignment. We target the direct household assignment contrast in these intervention villages, weighting villages equally, rather than the total village-policy contrast between the 80 villages. For a causal interpretation, we assume consistency for the specified Graduation package, a stable assignment environment, and no within-village interference in the simple response model. To allow spillovers, researchers would specify exposure averaging compatible with the feature independence conditions in addition to household randomization.

Endline monthly per-capita consumption is \path{ctotal_pcmonth_end}. We convert its raw numerical values to double precision before taking the natural logarithm. Of 1,609 households in the within-village randomized frame, 1,529 have observed consumption: 762 of 810 controls and 767 of 799 treated households. Missingness is therefore 48 controls (5.93\%) and 32 treated households (4.01\%). Every observed consumption value is positive; no zero or negative response is removed. All 40 villages have at least three observations per arm, with minima eight controls and seven treated households. No village is dropped by this outcome-availability check. To identify the native response-population contrast, we assume response-independent observation within village and arm, in addition to the independent-household model below. Under this condition and the independent-household model, we identify the native response-population contrast from the complete cases despite differential missingness.

The fixed additive basis is an intercept, log \path{control_h_population}, and the unchanged zero/one \path{control_h_dist_municipio}. The raw labels identify a village-population field and a municipality-distance indicator. We use the recorded fields as fixed descriptors: population units and the original measurement record remain unverified, and we retain the binary distance coding associated with the ambiguous label ``1--5 km'' without assigning a precise cutoff or treating it as continuous distance. Each variable has exactly one distinct nonmissing value in each of the 40 villages. Both have 27 missing household-level copies (two controls and 25 treated households), all with observed outcomes. We assign each household its village's recorded unique value and retain these households, using the observed village attribute without imputing from outcomes or excluding households. For a prospective interpretation as baseline context, researchers would verify the original measurement timing as well as the field definitions.

\paragraph{Exact point correction under the specified response model.}
Condition on the 40-village frame, recorded $W$, and observed arm counts. We use these recorded native descriptors directly; they require no household composition pilot and do not share an outcome-trained fit with the response features. The arm means $\theta_{aj}=\int m_{aj}\,d\Pi_j$ use the same native population of eligible households $\Pi_j$ within village $j$, without imposing a common reference across villages. We assume identically distributed response samples within village and arm, with these means, variances $\sigma_{aj}^2$, and $\tau_j=\theta_{1j}-\theta_{0j}$. We assume independence among distinct households within arms, between the two arm samples, and across villages. Write $s_{aj}^2$ for the sample variance with denominator $n_{aj}-1$. Then
\[
 Z_j=\bar Y_{1j}-\bar Y_{0j},\qquad
 \widehat v_j=s_{1j}^2/n_{1j}+s_{0j}^2/n_{0j},\qquad
 v_j=\sigma_{1j}^2/n_{1j}+\sigma_{0j}^2/n_{0j}.
\]
For each arm, form the product across distinct households,
\[
 U_{aj}=\frac{(\sum_iY_{iaj})^2-\sum_iY_{iaj}^2}
                  {n_{aj}(n_{aj}-1)},\qquad
 Q_j=U_{1j}+U_{0j}-2\bar Y_{1j}\bar Y_{0j}
     =Z_j^2-\widehat v_j.
\]
We verify the algebraic equality by comparing the two implementations, which agree within $1.3\times10^{-14}$. Proposition~1 gives $\E(Q_j)=\tau_j^2$ under the stated response model, while $\E(Z_j)=\tau_j$ and $\E(\widehat v_j)=v_j$. We obtain these point identities under finite second moments. For a fixed finite cohort, averaging over treatment randomization gives the additional subtraction $S_{\tau,j}^2/N_j$ stated in Section~\ref{sec:targetextensions}; we distinguish that target from the response-population effect square estimated here.

Let $\bm{X}$ be the fixed $40\times3$ additive design, $\bm{P}=\bm{X}(\bm{X}^\trans \bm{X})^{-1}\bm{X}^\trans$, $\bm{P}_0=\one\one^\trans/J$, and $\bm{C}=\diag(v_j)$, with $J=40$. The finite-site targets and estimators are
\begin{align*}
 B_F&=\bm{\tau}^\trans(\bm{P}-\bm{P}_0)\bm{\tau}/J,&
 \widehat B_F&=\{\bm{Z}^\trans(\bm{P}-\bm{P}_0)\bm{Z}-\tr[(\bm{P}-\bm{P}_0)\bm{\widehat C}]\}/J,\\
 R_F&=\bm{\tau}^\trans(\bm{I}-\bm{P})\bm{\tau}/J,&
 \widehat R_F&=\{\bm{Z}^\trans(\bm{I}-\bm{P})\bm{Z}-\tr[(\bm{I}-\bm{P})\bm{\widehat C}]\}/J,
\end{align*}
where $\bm{\widehat C}=\diag(\widehat v_j)$ and $V_F=B_F+R_F$. These are projection components for this recorded village basis, not the unrestricted population components $B$ and $R$.

For either $\bm{A}=\bm{P}-\bm{P}_0$ or $\bm{A}=\bm{I}-\bm{P}$, the uncorrected quadratic includes response-sampling noise:
\[
 \E\{\bm{Z}^\trans \bm{A}\bm{Z}/J\}=\bm{\tau}^\trans \bm{A}\bm{\tau}/J+\tr(\bm{A}\bm{C})/J.
\]
Indeed, put $\bm{Z}=\bm{\tau}+\bm{\varepsilon}$, use $\E(\bm{\varepsilon})=0$ and $\E(\bm{\varepsilon}\bm{\varepsilon}^\trans)=\bm{C}$, and take the expectation of the quadratic expansion. Both choices of $\bm{A}$ are positive semidefinite, so uncorrected OLS explained and residual sums contain nonnegative measurement-noise terms. Subtracting the unbiased trace terms proves the point-expectation assertion in Proposition~\ref{prop:finiteprojection}. The observed uncorrected total is .04745231, whereas $\widehat V_F=.00744680$; their difference is $(1-1/J)\overline{\widehat v}=.04000551$. We retain the signed corrected estimates; their unbiasedness is a point-estimation property, while the interval approximation is described below.

\paragraph{Working uncertainty for the primitive moments.}
We save the basis $\bm{H}_F\in\mathbb R^{J\times r}$ with first column $\one/\sqrt J$, followed by the left singular vectors of the centered, standardized specified $W$; it spans $\bm{X}$ and is fixed without using outcomes. Define $\bm{\gamma}=\bm{H}_F^\trans\bm{\tau}/\sqrt J$ and $q=\|\bm{\tau}\|^2/J$. Then
\[
 B_F=\sum_{k=2}^r\gamma_k^2,\qquad
 R_F=q-\|\bm{\gamma}\|^2,\qquad V_F=q-\gamma_1^2,
 \qquad q\ge\|\bm{\gamma}\|^2.
\]
We estimate the primitive moments by $(\bm{H}_F^\trans \bm{Z}/\sqrt J,\bar Q)$. Proposition~\ref{prop:independent_cell_covariance} gives the exact identities $\Cov(Z_j,Q_j)=2\tau_jv_j$ and $\Var(Q_j)=4\tau_j^2v_j+d_j$ under the independent response model with finite second moments, where
\[
 d_j=\frac{2\sigma_{0j}^4}{n_{0j}(n_{0j}-1)}+
     \frac{2\sigma_{1j}^4}{n_{1j}(n_{1j}-1)}+
     \frac{4\sigma_{0j}^2\sigma_{1j}^2}{n_{0j}n_{1j}}.
\]
Thus the exact primitive covariance has blocks $\bm{H}_F^\trans \bm{C}\bm{H}_F/J$, $\bm{H}_F^\trans(2\tau_jv_j)_j/(J\sqrt J)$, and $\sum_j(4\tau_j^2v_j+d_j)/J^2$. These identities hold for non-Gaussian responses as well. For interval construction, we substitute $Z_j$ and sample variances, retaining the off-diagonal primitive-moment covariances in the saved matrix. This plug-in calculation estimates the unknown covariance ingredients; it does not inherit the exact unbiasedness of the point correction. We use the marginal standard errors to construct a Bonferroni normal confidence rectangle with critical value 2.4977 for the four additive coordinates. The normal approximation concerns the primitive estimates, rather than the distribution of individual household responses.

For the reported containing bounds, let $[\ell_k,u_k]$ be the rectangle for $\gamma_k$ and $[\ell_q,u_q]$ the interval for $q$. Define $a_k=\min_{t\in[\ell_k,u_k]}t^2$ and $b_k=\max_{t\in[\ell_k,u_k]}t^2$. If $u_q\ge\sum_ka_k$, intersecting with $q\ge\|\bm{\gamma}\|^2$ yields the containing ranges
\begin{align*}
 B_F&\in[\sum_{k\ge2}a_k,\ \min\{\sum_{k\ge2}b_k,u_q-a_1\}],\\
 R_F&\in[\max\{0,\ell_q-\sum_kb_k\},\ u_q-\sum_ka_k],\\
 V_F&\in[\max\{0,\sum_{k\ge2}a_k,\ell_q-b_1\},\ u_q-a_1].
\end{align*}
We report containing ranges for the feasible projection. If the primitive intersection is empty, we report $[0,\infty)$ for each component; the main additive basis has a nonempty intersection. We calculate the component ranges using the unrestricted log-consumption scale and report the explanation ratio as a point estimate. We interpret the nominal 95\% label through the working normal approximation with estimated moment ingredients, small arm samples, and unadjusted historical selection, rather than finite-sample randomization inference. Because the rectangle and containing ranges can depend on a rotation of $\bm{H}_F$, we save the exact outcome-independent basis in \path{Native2W_additive_frozen_basis.csv}.

\paragraph{Basis sensitivity and reproducibility.}
When we add the single interaction between log population and the recorded distance indicator, the point explanation fraction changes from 62.14\% to 49.47\%, while maximum leverage rises from .209 to .548. We compare these fixed bases to assess sensitivity across different projection targets. We interpret the component ranges under this fixed reanalysis specification and report the explanation fraction descriptively, with historical application and endline screening unadjusted. We retain the main additive component estimates and uncertainty: the lower bound for $B_F$ is small, and the residual range includes zero.

The reproducibility files contain the household flow, native-attribute frame, arm means and variances, distinct-household products, fixed bases, leverage, covariance matrices, transformation code, and raw-file checksum.

\subsection{Finite-sample calibration with the Honduras village sizes}
\label{sec:honduras_smallJ_stress}
We assess the implemented primitive rectangle and containing ranges using the 40 recorded villages' fixed arm counts, saved additive basis, and observed arm sample variances as simulation inputs. We set the control mean to zero and define artificial village effects by
\[
 \tau_j^{(a)}=\overline Z+a(Z_j-\overline Z),\qquad a\in\{0,.35,1\}.
\]
We use the observed differences to specify these simulation means; the original trial effects remain unknown. Independent household errors have either a standard normal law, centered $\chi^2(3)/\sqrt6$, or $t(5)\sqrt{3/5}$, scaled to the recorded arm variances. Each setting uses 2,000 replications, with seed 20261006 and a single generator consumed in the saved order. We recompute $Z_j,Q_j$, sample variances, the four-coordinate rectangle, and the containing ranges in every replication using the application code's formulas.

\begin{table}[htbp]\centering\small
\caption{Calibration of the Honduras interval procedure under independent response sampling with fixed village sizes, 2,000 replications per setting. Primitive refers to simultaneous coverage of $(\bm{\gamma},q)$; components refers to simultaneous inclusion of $B_F,R_F,V_F$. The artificial residual targets are 0, .00500, and .04081 at $a=0,.35,1$, respectively. Width has squared log-consumption units.}
\label{tab:honduras_smallJ_stress}
\begin{tabular}{lrrrr}
\toprule
Error distribution & $a$ & Primitive & Components & Mean $R_F$ width\\
\midrule
Normal & 0 & .9650 & .9750 & .04979\\
Normal & .35 & .9575 & .9905 & .05773\\
Normal & 1 & .9635 & .9925 & .11064\\
Centered $\chi^2(3)$ & 0 & .9560 & .9740 & .05373\\
Centered $\chi^2(3)$ & .35 & .9545 & .9930 & .06219\\
Centered $\chi^2(3)$ & 1 & .9660 & .9945 & .11369\\
Standardized $t(5)$ & 0 & .9705 & .9770 & .05404\\
Standardized $t(5)$ & .35 & .9665 & .9920 & .06208\\
Standardized $t(5)$ & 1 & .9655 & .9945 & .11322\\
\bottomrule
\end{tabular}

\end{table}

We obtain primitive coverage from .9545 to .9705 and simultaneous component inclusion from .9740 to .9945. Every residual range contains its artificial target, and no primitive intersection is empty in these replications. For a reported inclusion proportion $p$, the Monte Carlo standard error is $\sqrt{p(1-p)/2000}$; it is at most .0047 for the primitive and joint-component proportions in the table. We obtain wide residual ranges, with mean widths .04979--.11369, and interpret their high inclusion rates together with their limited resolution for separating small positive residual variation from zero. We use this experiment to assess the interval implementation at the stated sizes and distributions, imposing outcome availability and household independence as part of the simulation model. We retain outcome availability and household independence as assumptions for interpreting the original data, distinct from the interval calibration reported here.

\end{document}